\documentclass[onecolumn]{article}
\usepackage{amsmath,amssymb,amsthm}
\usepackage{hyperref}
\usepackage{tikz}
\usepackage{bbm}
\usepackage[titletoc,toc,title]{appendix}
\usepackage{pgf}
\usetikzlibrary{arrows,shapes,positioning,backgrounds}
\usepackage[citestyle=authoryear, %-icomp,
    natbib=true,
    maxbibnames=9,
  backend=bibtex
  ]{biblatex}
\usepackage{anysize}
\theoremstyle{definition}
\newtheorem{theorem}{Theorem}[section] %% numbering on section level
\newtheorem{remark}{Remark}[section] %% dito
\newtheorem{definition}{Definition}[section]

\newtheorem{proposition}{Proposition}[section]
\newtheorem{example}{Example}[section] % [section]
\newtheorem{lemma}{Lemma}[section]

\numberwithin{equation}{section}

\newcommand{\length}[1]{\left\lvert #1\right\rvert}
\newcommand{\abs}[1]{\bigl\lvert #1\bigr\rvert}
\newcommand{\inv}[1]{#1^{-1}}
\newcommand{\set}[1]{\{#1\}}

\newcommand{\one}{{\mathbbm{1}}}
\newcommand{\real}{\mathbb{R}}

\newcommand{\complex}{\mathbb{C}} % complex numbers
\newcommand{\sd}{\,|\,} % ``so dass''
\newcommand{\goesto}{\rightarrow} % sequences 
\newcommand{\wo}{\backslash} % without
\newcommand{\struct}[1]{\mathcal{#1}} % struktur: shorter
\DeclareMathOperator{\Fix}{Fix}
\DeclareMathOperator{\SSLS}{SLS}

\title{
  Opinion dynamics and wisdom under %in-group bias and 
  out-group
  discrimination 
} 
\author{Steffen Eger\\ Computer Science Department\\ Technische
  Universit\"at Darmstadt}
\date{}

\begin{document}
\maketitle

\begin{abstract}
%We study distrust in social networks. 
%[Here goes the abstract].
We study a DeGroot-like opinion dynamics model in which agents may
oppose other agents. As an underlying motivation, in our setup, agents
want to adjust their 
opinions to match those of the agents %they follow
of their `in-group' and, in addition, they want to
adjust their opinions to match the `inverse' of those 
of the agents
of their %they oppose (their 
`out-group'. 
Our paradigm can account for 
 persistent disagreement in connected
societies %(Theorem \ref{cor:w-}) %and 
%both
%as well as  
%long-run consensus 
as well as 
bi- and multi-polarization.
% (Proposition
%\ref{prop:multipolarization}). 
Outcomes depend  
upon network %(`multigraph') 
structure and the choice of deviation function
modeling the mode of opposition between agents. 
For a particular choice of deviation function, which we call soft
opposition, we 
derive necessary and sufficient conditions for long-run polarization.
% (Theorem
%\ref{theorem:main}). 
We also consider
social influence (who are the
opinion leaders in the network?) as well as the 
question of wisdom in our 
na\"ive learning paradigm, finding that wisdom is difficult to attain
when there exist sufficiently strong negative relations between
agents.\footnote{Earlier and more verbose working paper versions of
  this article can 
  be found at \url{http://arxiv.org/pdf/1306.3134} and the
  author's personal website.}
%Psychologically and socio-economically, we
%interpret opposition as arising 
%from group identity structure (out-group discrimination) or, more particularly,% 
%from rebels/anti-conformists;
%countercultures; rejection of the norms and values of disliked others,
%as `negative referents'; or, simply, distrust. 
\end{abstract}

\section{Introduction}\label{sec:introductionOpposition}
On many %issues of everyday life, such as 
economic, political, social, and 
religious agendas, 
disagreement among individuals is %a 
pervasive. 
%phenomenon: 
For example, the following are or have been highly debated:
whether abortion, gay marriage, or death penalty should be
legalized or not;  
whether Iraq had weapons of mass destructions;  
%\footnote{See the polling data in ``Iraq: The Separate
%Realities...'' \parencite{Wpo2006}.} 
the
scientific standing of evolution; 
whether taxes/social subsidies/unemployment benefits/(lower bounds on)
wages should be 
increased or decreased;   
the right course of government in general; the effectiveness of
alternative (or `standard') medicine such as homeopathy.\footnote{Our
examples are, i.a., taken 
from \textcite{Abramowitz2005}, \citet{Acemoglu2011}, and \textcite{Golub2012}.
%We also note, however, that the `scope' of disagreement in society
%is disputed in the relevant literature, see, e.g., \textcite{Baldassari2007}. %But, of course, people clearly disagree,
%at least on society's `controversial' issues. 
} 
In fact, in 
  certain contexts  
such as the political arena, disagreement is `built into' and
essential part of the 
system of opinion exchange \parencite{Jones1995,Cohen2003}. %Moreover, 
%Conversely, 
Yet, contradicting this factual basis, 
it has been observed that the
phenomenon 
of disagreement %stands in contrast to the predictions of, in the
is not among the predictions of
%, in the 
%social and economic context, 
renown and
widely used theoretical %social and economic 
models of opinion dynamics in the social and economic context.\footnote{See, e.g., the discussions
\textcite{Acemoglu2011,Acemoglu:2013}. See
also \textcite{Abelson1964}.} 
%\parencite[see, e.g., the discussions
%in][]{Acemoglu2011,Acemoglu:2013,AcemogluSubmitted}). 
Namely, in 
these 
models, %it is usually predicted 
a standard prediction is 
that agents tend toward a
\emph{consensus opinion}, that is, that all agents eventually hold the
same opinion (or belief)\footnote{Typically, %we will use 
in the literature, 
the term \emph{belief} is used when %we think of 
there exists %ing 
a \emph{truth} for
an agenda, and %we will use 
the term \emph{opinion} is used when 
%we do not
%assume that such a true state of affairs must exist. However,
%occasionally,
truth is not explicitly modeled. Like our related work, we more
generally subsume under the term opinions also beliefs, judgements,
estimations, or even norms and values, depending on the application
scenario.} %, cf.\ \textcite{Buechel2015}.} 
%, although this may vary from author to
%author and discipline to discipline.} 
%In this work, where we only
%consider the latter situation, we typically say that agents hold 
%\emph{opinions} on issues, but 
%we take the freedom to occasionally use both terms interchangeably.}
about any specific 
issue. Typically, this  
%prediction generally 
applies to both (fully rational) Bayesian frameworks --- which is the reason
why \textcite{Acemoglu2011} call them 
``[no] natural framework[s] for 
understanding persisent disagreement'' (p. 6) --- and non-Bayesian
(boundedly rational) setups such as the famous DeGroot model of opinion
dynamics \parencite{DeGroot1974},  
where consensus obtains %under mild conditions (see below). 
as long as the social network wherein agents communicate with each
other is strongly 
connected (and aperiodic).
%\footnote{For a recent discussion of the
%`problem of consensus', see, e.g., \textcite{Acemoglu2011}; for an
%early discussion of the problem, 
%see, e.g., \textcite{Abelson1964}.}

Concerning the non-Bayesian DeGroot model, as we consider in this
work, a few amendments have more recently been suggested which are
capable of producing %the phenomenon of 
disagreement among agents. %On the one hand,
In one strand of literature, 
models including a \emph{homophily} mechanism, whereby agents limit
their communication to individuals whose opinions are not too
different from their own, can reproduce patterns
of opinion diversity and
disagreement \parencite{Deffuant2000,Hegselmann2002}.
In another strand, Daron Acemoglu and
colleagues \parencite[cf.][]{Acemoglu2011} introduce two 
types of agents, \emph{regular} and \emph{stubborn}, whereby the
latter never update their opinions but `stubbornly' retain their old
beliefs, which may be considered an autarky condition. Multiple
stubborn agents with distinct opinions on a certain agenda may then
draw society toward distinct opinion clusters. 
Such stubborn
agents, it is argued, may appear in the form of opinion leaders,
(propaganda) media, or political parties that wish to influence others
without receiving any feedback from them. 
%To induce disagreement,    
As a solution to the disagreement problem, however,  
both of these model types rely on a problematic `disconnectedness
condition', insofar as disagreement only obtains when 
%certain
%subclasses of agents are isolated from others in that 
there is no 
(uni- or bilateral)
information flow %between them.
between certain subclasses of agents.

In this work, we investigate an alternative explanation of
disagreement, which can also explain disagreement in connected
societies. We consider a non-Bayesian DeGroot-like opinion dynamics 
model in which agents are related to each other via \emph{two types of
  links}. One link type represents %the usual `weight' that one agent
the \emph{degree} or \emph{intensity} of relationship between agents
and is given by  
nonnegative real numbers. 
%places upon another in DeGroot learning models --- these weights, in
%DeGroot models, 
%typically represent `trust' between agents, importance, or simply a
%`listening/connectedness structure' and are given by real numbers,
%and, in our model, have the interpretation of strength or intensity of
%relationship between two agents ---
The other 
link type represents whether agents \emph{follow} or \emph{oppose} (\emph{deviate from}) each other,
that is, it represents the \emph{kind of
relationship} between agents. We assume that \emph{group
identity} causes agents to follow their \emph{in-group} members and to
deviate from their \emph{out-group} members. In-group favoritism and
out-group discrimination are important and well-established notions
  in social psychology (see, for
  instance, \cite{Tajfel:1971,Brewer1979,Castano2002}). They have
  also more 
recently been included in economists' models (e.g., in an 
experimental context, \cite{Charness2007,Ben-Ner:2009,Chen:2009,Hargreaves:2009,Fehrler2013,Tsutsui:2014}). Experimentally, it has been shown that even minimal group
identities, induced by a random labeling of groups, may lead to
intergroup discrimination. When group membership is more
salient, \textcite{Charness2007} show that there is much more cooperation
between in-group members than between out-group members in a
prisoner's dilemma game, and \textcite{Fehrler2013} show that individuals
associating with particular NGOs (non-governmental organizations)
strongly discriminate against out-group members (those that do not
associate with an NGO) in a trust
game.\footnote{In a `field' setting, the in-group/out-group
distinction may prominently be seen as arising, e.g., in a
(main stream) culture/counterculture 
(e.g., hippies, punks, etc.)
dichotomy  \parencite{Yinger1977} or in classical party divisions
(e.g., Republicans vs.\ Democrats) in
the field of politics, etc.}
%; see also \textcite{Hargreaves:2009} who likewise find negative
%discrimination against outsiders in trust games.
Analogously, %as a foundation principle of our model, 
we assume
that agents want to coordinate with their in-group members (have
negative utility from holding different opinions than in-group
members) and want to anti-coordinate with their out-group members
(have negative utility from not deviating from the opinions that their
out-group members hold).\footnote{Out-group discrimination
(opposition) is also closely related 
to what has been 
  termed \emph{rejection} of beliefs, actions, and values of
  dissimilar/disliked  
  others. According to this concept, agents change their normative
  systems to become more dissimilar to interaction partners they
  dislike
  (cf.\ \cite{Abelson1964,Kitts2006,Tsuji2002};
  cf.\ also \cite{Groeber2013}) insofar
  as disliked others may serve as `negative referents' who inspire
  contrary behavior. While in controlled
  experiments \textcite{Takacs2014} do not find strong evidence for
  the tenet that individuals disassociate from the 
  opinions of a disliked source, their study explicitly
  excludes a group identity structure. Moreover, as the authors argue,
  their laboratory experiment may
  have ``suppressed the emotional processes that in field settings
  induce disliking and rejection of others' opinions.''}
  A special case of our model is when an agent opposes everyone but
  himself, i.e., his in-group is himself and his out-group is all `the
  rest'. In some works, such agents have been referred to
  as \emph{rebels} or \emph{anti-conformists} in contrast
  to \emph{conformists} \parencite{Jackson2008,Cao2011,Javarone2014,Jarman2015}. 
  %\footnote{We also note that in social network theory, antagonistic
  %relationships between agents are nothing novel, with early work 
  %in this context dating back to the 1940's and 1950's (see
  %Chapter 5 in \cite{Beasley2010} and
  %references therein). Applications have ranged from international
  %relations (alliances vs.\ hostile relations) to trust/distrust
  %dichotomies --- recent empirical validity 
  %of both positive and negative relationships between individuals in
  %social 
  %networks is, amongst others, provided in \textcite{Leskovec2010}.
  %Often, the concept 
  %of \emph{signed networks} (network links have negative or positive
  %`signs') has been used to model both positive and 
  %negative influences. Novel in our context is the application of
  %these 
  %notions to the problem of \emph{opinion dynamics}, but see also our discussion
  %in Section \ref{sec:relatedWork}.} 
  %It has been claimed that such forces may
  %lead to the co-existence of multiple distinct norms, in contrast to
  %the evolution of a single norm among individuals. 

%Ignoring our novel group identity specification, 
Our model closely follows
the literature on learning through communication in a given social
network
(cf.\ \textcite{DeGroot1974,Demarzo2003,Golub2010,Acemoglu2010a,Buechel2015}). 
There, 
the standard assumption is that agents learn from others in a na\"ive
manner, %ignoring 
not properly accounting for 
the repetition of opinion signals, 
which \textcite{Demarzo2003} call `persuasion bias'. A now classical 
argument is that if all agents' initial beliefs/opinions were
independent and unbiased
estimates of the true value (of a discussion topic), then taking a
weighted average of the 
agents' beliefs in one's social network (where the weights are
proportional to the inverses of the beliefs' variances) is an optimal
aggregation strategy. Then, continuing to average 
--- %which is generally plausible, 
in order to incorporate more remote
information, e.g., from friends of friends --- \emph{in the same manner} is a
boundedly rational 
heuristic that treats the evolving information signals as novel,
not accounting for their cross-contamination. Such a heuristic aggregation of
opinion signals appears quite plausible given the processing costs 
involved in exact inference in this setup (cf.\ \cite{Golub2010}). Also,
recent experimental evaluations find that the na\"ive DeGroot model 
is a much better approximation of information aggregation in network
interactions than `fully rational' Bayesian approaches and that
individuals are indeed 
affected by persuasion bias \parencite{Corazzini2012}. In our model,
we posit that 
agents are subject to the same biases involving processing of
cross-contaminated information and are, \emph{in addition},
susceptible to \emph{in-group bias}, attempting to coordinate with
in-group members and to anti-coordinate with out-group
members.\footnote{If one wanted to construct an argument that closely
follows that of \textcite{Demarzo2003}, one might posit that, in our
model, agents `correct' their out-group members' opinion signals ---
possibly because of distrust --- before averaging.}

This work is structured as follows. Section \ref{sec:model} presents
the model and  
%Section \ref{sec:examples} 
gives two introductory examples. 
Sections \ref{sec:discrete} and \ref{sec:distrust} 
%and \ref{sec:wisdom} 
present our main
results, on persistent disagreement (Theorem \ref{cor:w-}) and bi- and
multi-polarization (Proposition \ref{prop:multipolarization}). For a
special case of our model, we derive necessary and sufficient
conditions for long-run polarization (Theorem \ref{theorem:main}) as well as
opinion leadership (Theorem \ref{theorem:social_influence}) as further
main results. 
In Section \ref{sec:conclusions}, we summarize and conclude with a
discussion on wisdom. To make this
work more or less self-contained, we provide concepts, e.g., from graph
 and matrix theory in the
appendix, to which we also relegate all our proofs.

%\section{Related Work}\label{sec:relatedWork}
%\input{related}

\section{Model}\label{sec:model}
Let $S$ be a finite set (`discrete model') or a subset of the real
numbers (`continuous model'), which we
refer to as 
\emph{opinion spectrum}.\footnote{For the continous DeGroot model as
  we discuss, $S$ is typically modeled as a convex subset of the real
  numbers, that is, $\sum_{j}   
\alpha_jx_j\in S$ for all finite numbers of elements $x_j\in S$ and
all weights $\alpha_j\in 
      [0,1]$ such that $\sum_j \alpha_j = 1$. 
      %Below, we will usually
      For convenience, we 
      think of $S$ as the whole of $\real$ or of some (closed and bounded)
      interval $[\alpha,\beta]$ for $\alpha\le\beta$.}  
Let $n\ge 1$ and let $[n]=\set{1,\ldots,n}$ be a set of $n$
agents. Consider the \emph{normal form game}
$([n],S_1\times\cdots\times S_n,U)$, where 
\begin{itemize}
  \item $[n]$ is the set of \emph{players},
  \item $S=S_1=\cdots=S_n$ is the \emph{action set} of each player,
  \item and $U=(u_1,\ldots,u_n)$, where $u_i:S^n\goesto\real$ is the
    \emph{payoff/utility} 
    function of player $i\in[n]$. 
\end{itemize} 
Let $\mathfrak{F}$ be the identity function on $S$ --- that is,
$\mathfrak{F}(x)=x$ for all $x\in S$ --- and
let $\mathfrak{D}$ be a function $\mathfrak{D}:S\rightarrow S$
that is not the identity function, and which
    we term \emph{deviation function} ({in the most general form of}
  our model, 
    %we subscript $\mathfrak{D}$ as in $\mathfrak{D}_{ij}$). 
we let deviation functions $\mathfrak{D}$ depend on particular
agents $i$ and $j$ involved, that is, we subscript $\mathfrak{D}$ as
in $\mathfrak{D}_i$ or $\mathfrak{D}_{ij}$).  
We assume that agents are connected 
via a (social) network $\mathbf{W}$, where $W_{ij}\ge 0$ denotes the
strength of relationship between 
agents $i$ and $j$.\footnote{Throughout, we denote the entries of a vector
$\mathbf{u}$ as $u_i$ or $[\mathbf{u}]_i$ and analogously for
matrices.}
More precisely, $W_{ij}$ signals the degree of importance of agent $j$ for
agent $i$, and we do not require $\mathbf{W}$ to be symmetric, that
is, $W_{ij}$ and $W_{ji}$ may differ.  
Assume that %$\mathbf{W}$ is not necessarily symmetric, that  
$W_{ii}=0$ and $\sum_{j=1}^n W_{ij}=1$, for all
$i\in[n]$.\footnote{The assumption $W_{ii}=0$ can be relaxed, see
\textcite{Groeber2013}. In subsequent sections, we do not always assume that
$W_{ii}=0$.} Assume
further that
player $i$ has payoff for action profile $(b_1,\ldots,b_n)$
\begin{align}\label{eq:utility}
  u_i(b_1,\ldots,b_n) = -\sum_{j\in\text{In}(i)}W_{ij}(b_i-\mathfrak{F}(b_j))^2-\sum_{j\in\text{Out}(i)}W_{ij}(b_i-\mathfrak{D}(b_j))^2
\end{align}
for the continuous model. %or, 
Here, $\text{In}(i)\subseteq[n]$ is the \emph{in-group} (set of \emph{friends})
%(\emph{in-group}) 
of
player $i$ and $\text{Out}(i)\subseteq[n]$ is the \emph{out-group} (set of
%\emph{enemies} (\emph{out-group}) 
\emph{enemies}) 
of player $i$. 
For the discrete model, assume that the analogous payoff is 
\begin{align}\label{eq:utility_discrete}
    \begin{split}
  u_i(b_1,\ldots,b_n) &= %-\sum_{j: i \text{ trusts } j}
  %W_{ij}(b_i-b_j)^2-\sum_{j: i \text{ distrusts } j}
  %W_{ij}(b_i+b_j)^2,  
  %-\sum_{j=1}^n W_{ij}(1-\one(F_{ij}(b_j),b_i)) \\ 
  %&= -\sum_{j=1}^n W_{ij}(1-\one_{F_{ij}(b_j)}(b_i)) \\ 
  -\sum_{j: j\in \text{In}(i),b_i\neq \mathfrak{F}(b_j)} W_{ij}-\sum_{j:
    j\in\text{Out}(i), b_i\neq \mathfrak{D}(b_j)} 
  W_{ij}. 
  \end{split}
\end{align}
%where we let $\one(r,t)=1$ if $r=t$ and zero
%otherwise.
Utility functions
$u_i$ in 
Eq.\ \eqref{eq:utility} and Eq.\ \eqref{eq:utility_discrete} say that
player $i$ has disutility from 
choosing a different action than his friends and has disutility from
choosing a different action than the `opposite' action of his enemies,
where
$\mathfrak{D}$ specifies what the opposite of an action is. 

When each
agent repeatedly plays a best response to the actions --- which in
our setup are opinions --- of the other
players, i.e., $i\in[n]$ chooses action $b_i$ that maximizes
$u_i(\cdot)$, then, in the continuous model, opinions 
evolve over time according  
to the following weighted average of (possibly, via $\mathfrak{D}$, 
`inverted') past opinions:   
\begin{align}\label{eq:updating}
  b_i(t+1) =
  \sum_{j\in\text{In}(i)}W_{ij}b_j(t)+\sum_{j\in\text{Out}(i)}W_{ij}\mathfrak{D}(b_{j}(t))
  = 
    \sum_{j=1}^n W_{ij}F_{ij}(b_j(t)),
\end{align}
 for $t=0,1,2,\ldots$, starting from some particular initial actions $b_i(0)$. %and according
Here, $F_{ij}\in\set{\mathfrak{F},\mathfrak{D}}$, depending on whether $j$
  is in $i$'s in-group or out-group, respectively. 
For the discrete model, the analogous best response action is weighted
majority voting of agents' (possibly inverted) past opinions: 
%to
\begin{align}\label{eq:discrete}
    b_i(t+1) = \text{arg}\max_{s\in S}\: \left\{\sum_{j=1}^n W_{ij}\one\bigl(F_{ij}(b_j(t)),s\bigr)\right\},
\end{align}
%in the discrete model, 
where %$\mathbf{F}$ is an $n\times n$
  %`matrix' with
 $\one(r,t)=1$ if $r=t$ and zero otherwise.
% Here, $F_{ij}\in\set{\mathfrak{F},\mathfrak{D}}$, depending on whether $j$
%  is in $i$'s in-group or out-group, respectively.
%Here,
Casting the updating processes \eqref{eq:updating} and \eqref{eq:discrete} in 
  a more compact  `operator' notation, we 
  write ($\mathbf{F}$ being the $n\times n$
  `matrix' with entries 
  $F_{ij}$) 
  \begin{align}\label{eq:matrix}
    \mathbf{b}({t+1}) = (\mathbf{W}\odot\mathbf{F})(\mathbf{b}(t)).
  \end{align}
  %where \footnote{For short, we will usually write $(\mathbf{W}\odot\mathbf{F})\mathbf{b}$
  %  instead of $(\mathbf{W}\odot\mathbf{F})(\mathbf{b})$.}
  Here, we let 
  the `operator' $\mathbf{W}\odot\mathbf{F}$
  act on a vector 
  $\mathbf{b}\in S^n$ in the manner prescribed in \eqref{eq:updating}
  and \eqref{eq:discrete},
  i.e., $\bigl[(\mathbf{W}\odot\mathbf{F})(\mathbf{b})\bigr]_{i}\overset{\text{def}}{=}
  \sum_{j=1}^n W_{ij}\cdot F_{ij}(b_{j})$ and analogously for the
  discrete model.\footnote{For short, we will usually write $(\mathbf{W}\odot\mathbf{F})\mathbf{b}$
    instead of $(\mathbf{W}\odot\mathbf{F})(\mathbf{b})$.}
  Equation \eqref{eq:matrix} may again be rewritten as  
  \begin{align}\label{eq:matrix2}
    \mathbf{b}({t}) = (\mathbf{W}\odot \mathbf{F})^t(\mathbf{b}(0)),
  \end{align}
  by which we denote the $t$-fold application of operator
  $\mathbf{W}\odot{\mathbf{F}}$ on $\mathbf{b}(0)$, that is,
  ${f}^t(\mathbf{b})=f(\cdots f(f(\mathbf{b})))$, where
  $f=\mathbf{W}\odot\mathbf{F}$. In the sequel, we refer to
  $\mathbf{W}\odot\mathbf{F}$ as `operator' or `social
  network'. 
  \begin{remark}
    In case $\mathbf{F}$ is the $n\times n$ matrix of identity
    functions, updating process \eqref{eq:matrix2} collapses to the
    standard DeGroot learning model where
    $(\mathbf{W}\odot\mathbf{F})^t$ is simply the $t$-th matrix power
    of matrix $\mathbf{W}$. %Put differently, 
    From an alternative (equivalent) viewpoint, 
    our model generalizes the
    standard DeGroot model insofar as the latter posits that
    $\text{Out}(i)=\emptyset$ for all $i\in[n]$. 
  \end{remark}
  %\begin{remark}
  %\end{remark}

We note that since the operator
$\mathbf{W}\odot\mathbf{F}$ in opinion 
updating process \eqref{eq:matrix} 
%which has
%$\mathbf{b}(t+1)=(\mathbf{W}\odot\mathbf{F})(\mathbf{b}(t))$, 
retrieves best responses of agents
to an opinion profile  
$\mathbf{b}(t)$, under utility functions
$u_i(\cdot)$ as 
in %\eqref{eq:best-response}, 
\eqref{eq:utility} or \eqref{eq:utility_discrete}, 
the fixed-points of
$\mathbf{W}\odot\mathbf{F}$ --- that 
is, the points $\mathbf{b}$ such that
$(\mathbf{W}\odot\mathbf{F})(\mathbf{b})=\mathbf{b}$ --- are 
%, by
%definition of a Nash equilibrium, 
the Nash equilibria of the normal form
games $([n],S^n,U(\cdot))$. Namely, for each such a fixed-point, all
players in $[n]$ play best responses to the other players' actions
(opinions). 
%In the subsequent
%sections, we pursue the task of finding
%$(\mathbf{W}\odot\mathbf{F})$'s fixed-points in more detail. 

We now illustrate our model with two examples, outlining its
relationship to other models discussed in the literature and hinting
at its potential for 
long-run disagreement. 
 \begin{example}\label{example:anticonformity}
    Let $S=\set{0,1}$ be a binary opinion space. Assume that 
    \begin{align*}
      \mathfrak{D}(x) = \begin{cases}
        1 & \text{if } x=0,\\
        0 & \text{if } x=1.
        \end{cases}
    \end{align*}
    Let $W_{ij}=\frac{1}{\length{N(i)}}$, where $N(i)$ denotes the set
      of \emph{neighbors} of agent $i$ in the social networks, i.e.,
      the set of agents $j$ for which $W_{ij}>0$. When 
      $F_{ij}=\mathfrak{F}$ for all $j\in[n]$, and agent $i$ updates
      opinions according to \eqref{eq:discrete}, then, at each time step,
      agent $i$ chooses the majority opinion among his neighbors'
      opinions. Such 
      individuals have also been called `conformists' in some
      contexts; e.g.,
      \textcite{Jackson2008,Cao2011,Javarone2014,Jarman2015}. Conversely,
      when 
      $F_{ij}=\mathfrak{D}$ for all 
      $j\in[n]$, and agent $i$ updates 
      opinions according to \eqref{eq:discrete}, then, at each time step,
      agent $i$ chooses the \emph{minority} opinion ($=$ majority of
      inverted opinions) among his
      neighbors' opinions. Such individuals have also been called 
      `anti-conformists' or `rebels'. When weights are
      non-uniform and/or agents follow some of their peers while
      deviating 
      from others, then the current setup may yield interesting
      generalizations of the basic conformist/non-conformist model.\footnote{
      It is also worthy to note that binary opinion spaces (in a
      conformist/anti-conformist setup) are closely related to %binary action
      \emph{games on networks} (cf.\ \cite{Jackson2008} and
      references therein) 
      with binary action spaces. In a society with only
      anti-conformists on simple graphs (undirected graphs with no
      self-loops), \emph{maximally independent sets} $S\subseteq[n]$
      of agents for 
      which it holds that $W_{i,S}=\sum_{j\in S}W_{ij}>1/2$ for all
      $i\in[n]\wo S$ form 
      pure strategy Nash equilibria of the underlying games 
      in the sense that assigning one action/opinion to all
      agents in $S$ and the complementary action/opinion to all agents
      in $[n]\wo S$ constitutes a setting where each agent plays a
      best response to the actions/opinions of the others. In a
      society with both conformists and anti-conformists, (pure strategy) Nash
      equilibria exist on networks in which each agent assigns weight
      mass $>1/2$ to conformists: an equilibrium is where all the
      conformists take one action (hold one opinion), and all the
      anti-conformists the other (cf.\ \cite{Jackson2008}, p.272).}
      
  \end{example}
    \begin{example}\label{example:complexSociety}
    %To make a more complex example, let $S$ be the interval
    %$[-1,1]$. For illustration, we may think of the opinion $x=-1$ as
    In this example, 
     we let $S=[-1,1]$ and think of the opinion $x=-1$ as 
    extreme left-wing opinion, $x=+1$ as extreme right-wing opinion
    and of opinions $-1<x<1$ as more moderate opinions ($x=0$ as
    `center' opinion). Assume there are six individuals,
    organized in four
    groups ${{A}}$, ${{B}}$, ${{C}}$,
    ${{D}}$, members of each of which follow members of 
    their own group  
    and deviate from the members of the other groups. %Assume that
    Hence, let 
    $[n]=\set{{1},{{2}},{{3}},{{4}},{{5}},{{6}}}$, $A=\set{1,2}, B=\set{3}, C=\set{4,5}, D=\set{6}$, 
      and
    \begin{align*}
      &\text{Out}(1)=\text{Out}(2)=\set{3,4,5,6},\:
      {\text{Out}(3)}=\set{1,2,4,5,6},\\
      & %\:
      {\text{Out}(4)}=\text{Out}(5)=\set{1,2,3,6},\:{\text{Out}(6)}=\set{1,2,3,4,5}.
    \end{align*}
    Assume the following specification of deviation functions for
    members in 
    each
    group of agents: 
    \begin{equation}
      \label{eq:devExample}
      \begin{split}
      &\mathfrak{D}_{A}(x)=1,\quad 
      \mathfrak{D}_{B}(x)=-1,\quad
      \mathfrak{D}_{C}(x)=\frac{x}{2},\\
      %\quad 
      &{\mathfrak{D}_{D,AB}(x)}=-x,\:\:\:\:
      \mathfrak{D}_{D,C}(x)=\text{sgn}(x)\sqrt{\length{x}}, 
      \end{split}
    \end{equation}
    for all $x\in[-1,1]$. Put differently, 
    agents in group ${A}$ ignore the actual opinion signals of members
    of their out-groups, simply interpreting any uttered opinion of an
    out-group individual as evidence of the opinion $1$.
    %i.e., 
    %\begin{align*}
    %  {D}_1(x) = 1, \text{ for all } x\in [-1,1]. 
    %\end{align*}
    Similarly, 
    agents in group ${B}$ interpret any opinion signal
    uttered by an out-group agent as evidence for opinion $-1$.
    %\begin{align*}
    %  {D}_2(x) = -1, \text{ for all } x\in [-1,1]. 
    %\end{align*}
    Agents in group ${C}$ are moderate
    in that 
    they `discount'
    (extreme) viewpoints that their out-group members hold.
    %Their
    %deviation function is given by
    %\begin{align*}
    %  {D}_3(x) = \frac{x}{2}, \text{ for all } x\in [-1,1].  
    %\end{align*}
    %and 
    Finally, agents in group ${D}$ %generally distrust opinions of
    %out-group members. They 
    more literally {invert} the opinions of members of (their out-)groups
    $A$ and $B$ --- possibly knowing of these agents' predispositions for
    extreme opinions of particular kinds.
    %via
    %\begin{align*}
    %  {D}_{4,A} = -x, \text{ for all } x\in [-1,1],
    %\end{align*}
    Moreover, they map the opinions of members of (their out-)group
    $C$ 
    %--- 
    %via 
    %\begin{align*}
    %  {D}_{4,B} = \text{sgn}(x)\sqrt{\length{x}}, \text{ for
    %    all } x\in [-1,1],
    %\end{align*}
    %which maps opinions $x$ 
    to \emph{more extreme opinions} for any value of
    $x$ between 
    $-1$ and $1$
    --- possibly 
    knowing of these agents' biases toward moderate %conservative
    opinions. 
%Assume that, for illustration,
According to this specification, matrix $\mathbf{F}$ looks as follows: 
   \begin{align*}
      \mathbf{F}=\begin{pmatrix}
        \mathfrak{F} & \mathfrak{F} &  \mathfrak{D}_{A} & \mathfrak{D}_{A}
        &  \mathfrak{D}_{A} & \mathfrak{D}_{A}\\
        \mathfrak{F} & \mathfrak{F} &  \mathfrak{D}_{A} &
        \mathfrak{D}_{A} 
        &  \mathfrak{D}_{A} & \mathfrak{D}_{A}\\
        \mathfrak{D}_{B} & \mathfrak{D}_{B} &  \mathfrak{F} & \mathfrak{D}_{B}
        &  \mathfrak{D}_{B} & \mathfrak{D}_{B}\\
        \mathfrak{D}_{C} & \mathfrak{D}_{C} &  \mathfrak{D}_{C} & \mathfrak{F}
        &  \mathfrak{F} & \mathfrak{D}_{C}\\
        \mathfrak{D}_{C} & \mathfrak{D}_{C} &  \mathfrak{D}_{C} & \mathfrak{F}
        &  \mathfrak{F} & \mathfrak{D}_{C}\\
        \mathfrak{D}_{D,AB} & \mathfrak{D}_{D,AB} &  \mathfrak{D}_{D,AB} & \mathfrak{D}_{D,C}
        &  \mathfrak{D}_{D,C} & \mathfrak{F}\\
      \end{pmatrix}.
    \end{align*}

   %\begin{figure}
   %  \begin{center}
   %  \input{pics/example3-2.tex}
   %  \end{center}
   %\end{figure}

      In Figure \ref{fig:generalExample}, we plot sample opinion
      dynamics 
      %when agents update their opinions according to
      %\eqref{eq:updating}, via $\mathbf{F}$ and $\mathbf{W}$ as
      %described, 
      when agents start with the initial consensus
      $\mathbf{b}(0)=(\mu,\ldots,\mu)^\intercal$, where $\mu=1/4$, and
      for an arbitrarily selected positive row-stochastic matrix
      $\mathbf{W}\in\real^{6\times 6}$. Note 
      how, in this case, agents' opinions polarize into extreme and
      moderate viewpoints and note how the opinion dynamics process
      (apparently) converges and stabilizes as time progresses. In the
      same figure, we also sketch the deviation functions defined in
      Equation \eqref{eq:devExample}. 
      %Also
      %note that no agent becomes wise, although all agents' initial
      %beliefs are accurate (i.e., are exactly truth $\mu$). 
      \begin{figure}[!htb]
        \begin{center}
          \input{pics/exampleNew.tex}
        \input{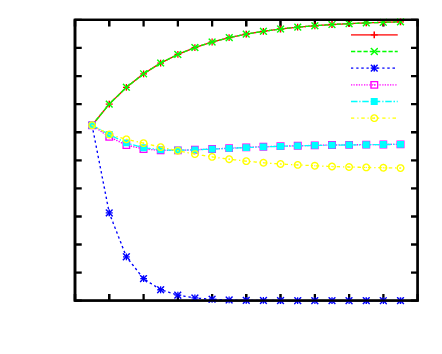}
        \end{center}
        \caption{Left: Deviation functions %for
          $\mathfrak{D}_{A}$,
          $\mathfrak{D}_{B}$,
          $\mathfrak{D}_{C}$, and
          $\mathfrak{D}_{D,AB}$,
          $\mathfrak{D}_{D,C}$ specified in the
          Equation \eqref{eq:devExample}, as well as identity function
          $\mathfrak{F}$, on $S=[-1,1]$.
          Right: Opinion dynamics for the society described in the example, starting from the initial consensus 
        $(1/4,\ldots,1/4)$ at time $0$. Agents $1$ and $2$ belong to
          group $A$, agent $3$ makes up group $B$, agents $4$ and $5$
          belong to group $C$, and agent $6$ makes up group
          $D$. Weights $W_{ij}$ set to arbitrary positive values.}
        \label{fig:generalExample}
      \end{figure}
  \end{example}

%\section{Examples}\label{sec:examples}
%\input{examples}

%\section{Further definitions and notation}\label{sec:preliminaries}
%\input{notation}

\section{Persistent disagreement, bi- and multipolarization}\label{sec:discrete}
In this section, we consider our extended DeGroot model in the abstract
situation of arbitrary deviation functions $\mathfrak{D}$. Our results
in this context concern 
%(a) the potential consensus Nash equilibria of
%the normal 
%form games set up as motivation for our extended DeGroot learning and,
%very relatedly, 
%(b) 
the possible consensus opinions that agents may hold in our model in
the long-run (Theorem
\ref{cor:w-}). We will find that %the consensus Nash equilibria and 
the long-run
consensus vectors can be determined quite simply: under appropriate
weight conditions, a certain consensus opinion vector $(c,\ldots,c)$ is an equilibrium if
and only if $c$ is a ($\mathfrak{D}$-)\emph{neutral} opinion with
respect to each   
agent's deviation function $\mathfrak{D}$, i.e., $\mathfrak{D}$ maps
the opinion $c$ to itself. 
This is our \emph{persistent
  disagreement} result: 
 as long as there exist (sufficiently strong) out-group relations between
 agents, society will disagree forever if there
 is no opinion which is neutral for each agent, no matter the agents' initial
 opinions. 
After investigating
 long-run persistent disagreement, we consider a particular form of
 disagreement, namely, bi- and multi-polarization (for abstract and
 arbitrary deviation functions $\mathfrak{D}$). 

      To begin with our formal analysis, we define 
      a few concepts.
      Let 
  $\struct{C}$ be the set of \emph{consensus opinion vectors} in $S^n$, i.e.,
  \begin{align*}
        \struct{C}=\set{(a_1,\ldots,a_n)\in S^n\sd
    a_1=\cdots=a_n}. 
    \end{align*}
      %\end{definition}
        Let $Y$ be an arbitrary set and 
        let   $Q$ be an arbitrary function $Q:Y\goesto Y$. 
        By $\text{Fix}(Q)$, we denote the set of \emph{fixed-points} of $Q$,
       that is, the set of all $x\in Y$ such that $Q(x)=x$. For a
      deviation function $\mathfrak{D}$, we call an opinion $x\in S$
      for which 
  $\mathfrak{D}(x)=x$ \emph{($\mathfrak{D}$-)neutral}. 
  %If $\text{Fix}(\mathfrak{D})=\emptyset$ (i.e., there are no
  %    $\mathfrak{D}$-neutral opinions), we call
  %    $\mathfrak{D}$ \emph{radical}.  
  %We call $a,b\in S$ \emph{($\mathfrak{D}$-)opposing viewpoints} if $\mathfrak{D}(a)=b$ and
  %    $\mathfrak{D}(b)=a$. 
  Let $A\subseteq[n]$ be an arbitrary subset of the set of agents and let 
  $i\in[n]$ be a particular agent. We denote by $W_{i,A}:=\sum_{j\in
  A}W_{ij}$ the total weight mass $i$ assigns to group $A$. 
%\end{definition}

%From fixed-points of $\mathbf{W}\odot\mathbf{F}$, we may turn to
%long-run dynamics in case that $\mathbf{W}\odot\mathbf{F}$ represents
%a continuous operator, or, alternatively, when it refers to our
%discrete updating model. This is captured in the following result. 
%\begin{definition}
%        For an arbitrary social network $\mathbf{W}\odot\mathbf{F}$,
%        let $\text{Lim}(\mathbf{W}\odot\mathbf{F})$ denote the set of
%        opinion profiles $\mathbf{b}\in S^n$ that may result in the
%        limit of opinion updating
%        process \eqref{eq:matrix2}. Formally:  
%        \begin{align*}
%        \text{Lim}(\mathbf{W}\odot\mathbf{F}) = 
%\set{\mathbf{b}\in S^n\sd
%  \mathbf{b}=\lim_{t\goesto\infty}(\mathbf{W}\odot\mathbf{F})^t\mathbf{b}(0),\text{
%    for some } \mathbf{b}(0)\in S^n}. 
%        \end{align*}
%\end{definition}
%      Our next main result, although formally a corollary to
%      Proposition \ref{prop:w-} and Lemma \ref{lemma:w-},
%      characterizes long-run consensus vectors in terms of
%      fixed-points of the functions $\mathfrak{D}_i$.  
      \begin{theorem}\label{cor:w-}
      Let $\mathbf{W}\odot\mathbf{F}$ be an arbitrary social
        network such that
      $F_{ij}\in \set{\mathfrak{F},\mathfrak{D}_i}$, for all
      %where deviation function $\mathfrak{D}_i$ corresponds to agent
      $i,j\in[n]$. Assume that either $\mathbf{W}\odot\mathbf{F}$ refers
        to the discrete model or that each function $F_{ij}$ in
        $\mathbf{F}$ is continuous. 
        Let $A=\set{i\in[n]\sd W_{i,\text{Out}(i)}>C_0}$
      denote the set of agents whose weight mass assigned to out-group
      members exceeds a particular threshold $C_0$; in the discrete
      case, $C_0=\frac{1}{2}$, and in the continuous case,
      $C_0=0$.
        Then: 
  \begin{align*}
    P_1[\text{Lim}(\mathbf{W}\odot\mathbf{F})\cap\struct{C}]
    = \bigcap_{i\in A}\text{Fix}(\mathfrak{D}_i),
  \end{align*}
   where $\text{Lim}(\mathbf{W}\odot\mathbf{F}) = 
\set{\mathbf{b}\in S^n\sd
  \mathbf{b}=\lim_{t\goesto\infty}(\mathbf{W}\odot\mathbf{F})^t\mathbf{b}(0),\text{
    for some } \mathbf{b}(0)\in S^n}$ denotes the set of
        opinion profiles $\mathbf{b}\in S^n$ that may result in the
        limit of opinion updating
        process \eqref{eq:matrix2}. 
        Moreover, $P_1$ projects consensus vectors
      $(c,\ldots,c)\in S^n$ on their first coordinate $c\in S$. If $A$
      is the empty set, we let $\bigcap_{i\in
      A}\text{Fix}(\mathfrak{D}_i)=S$. 
  %In particular, if one $\mathfrak{D}_i$ is radical, %$\text{Fix}(D)=\emptyset$ and
  %then 
  %\eqref{eq:matrix2} never converges to a consensus. 
\end{theorem}
%\begin{remark}
%        Corollary \ref{cor:w-} is not true, in general, if not all
%        functions $\mathfrak{D}_i$ are continuous. To illustrate,
%        consider Example \ref{example:probinv2} with deviation
%        function 
%        \begin{align*}
%        \tilde{\mathfrak{D}}(x) = 
%        \begin{cases}
%        -x&\text{if } x\neq 0,\\
%        1&\text{if } x=0.
%        \end{cases}
%        \end{align*}
%        in place of deviation function $\mathfrak{D}(x)=-x$. 
%        Then $\text{Fix}(\tilde{\mathfrak{D}})=\emptyset$, but
%        $(0,0)\in \text{Lim}(\mathbf{W}\odot\mathbf{F})\cap\struct{C}$. (In
%        fact, each initial opinion vector $\mathbf{b}(0)$ leads to the
%        limiting opinion vector $(0,0)$ under updating
%        process \eqref{eq:matrix2} and the social network
%        $\mathbf{W}\odot\mathbf{F}$ of
%        Example \ref{example:probinv2}, even with the new deviation
%        function $\tilde{\mathfrak{D}}$).  
%\end{remark}
%The implications of Corollary \ref{cor:w-} are manifold. Under the
%assumptions of the theorem, 
\begin{remark}
One of the implications of Theorem \ref{cor:w-}
is \textbf{persistent disagreement}, under 
%rather mild conditions and
the assumptions of the theorem
and further rather mild conditions, as
a prediction of our generalized DeGroot updating
process \eqref{eq:matrix2}. Namely, in particular, the relation $\bigcap_{i\in
        A}\text{Fix}(\mathfrak{D}_i)=\emptyset$ follows when, for
        instance:
\begin{itemize}
        %\item 
        %the possible consensus limit vectors that updating
        %process \eqref{eq:matrix2} may converge to can easily be found
        %by looking at the fixed points of the agent-specific deviation
        %functions $\mathfrak{D}_i$.
        %{\bf Persistent disagreement} under relatively mild
        %conditions: 
        %If an agent has
        %$\text{Fix}(\mathfrak{D}_i)=\emptyset$ or if two agents  
    
        %\begin{itemize}
                \item One agent's deviation function is \emph{radical}:
        $\text{Fix}(\mathfrak{D}_i)=\emptyset$. 
        \item Two agents' %do not share 
        assessment of what constitutes
                a neutral opinion differs:
                $\text{Fix}(\mathfrak{D}_{i})\cap\text{Fix}(\mathfrak{D}_j)=\emptyset$. 
        %\item 
        More generally, persistent disagreement follows whenever the
                agents in society have no common interpretation of
                neutrality: there exists no opinion $c\in S$ such that
                $c$ is $\mathfrak{D}_i$-neutral for all agents $i$. 
                %In general, we would
                %expect this to depend both on $S$ as well on society. 
        %\end{itemize}
\end{itemize}
%\begin{remark}
  Moreover, concerning deviation functions $\mathfrak{D}$, our only assumption
  was that there are points which they do not fix. Modeling
  deviation, however, a plausible (stronger) restriction on
  $\mathfrak{D}$ is that 
  $\mathfrak{D}(x)\neq x$ for 
  many, most, or all $x\in S$. Thus, potential long-run agreement,
  that is, $\bigcap_{i\in 
        A}\text{Fix}(\mathfrak{D}_i)\neq\emptyset$, would be
  particularly difficult to obtain, because this consists of a
  condition that is unlikely on the level of individual agents, and in
  addition contains a \emph{cross-agent} constraint that deviation
  functions would have to fix the \emph{same} point(s) $x\in S$ for all
  agents. 
\end{remark}
\begin{example}
        %Consider the complex example.
        We apply Theorem \ref{cor:w-} to the examples discussed
        previously. In Example \ref{example:anticonformity},
        $\text{Fix}(\mathfrak{D})=\emptyset$, so in a
        conformist/anti-conformist society, consensus cannot ensue if
        at least one agent opposes more than half of his social
        network. 

        In Example \ref{example:complexSociety},
        $\text{Fix}(\mathfrak{D}_{A})=\set{1}$,
        $\text{Fix}(\mathfrak{D}_{B})=\set{-1}$,
        $\Fix(\mathfrak{D}_{C})=\Fix(\mathfrak{D}_{D,AB})=\set{0}$, and $\Fix(\mathfrak{D}_{D,C})=\set{0,1,-1}$. While 
        Theorem \ref{cor:w-} does not directly apply to
        Example \ref{example:complexSociety}, since agent $6$'s deviation functions
        vary across out-group agents, we can nonetheless apply the
        theorem to the society consisting of agents $1$ through $5$
        and conclude that reaching a consensus is not possible for 
        this subsociety since, e.g., $\emptyset
        = \set{1}\cap\set{-1}$ (and assuming that respective weights
        satisfy the positivity assumption outlined in the theorem).
        Hence, since agents $1$ to $5$ cannot reach a consensus, then
        also agents $1$ to $6$ --- the overall society in the example
        --- cannot reach a consensus. 

        %$\Fix(\mathfrak{D})=\set{0}$, so society can settle on the
        %consensus $0$ --- in fact, in Figure \ref{fig:_evol}, 
        %%such a
        %%consensus is apparently reached, in the limit of repeated
        %%updating. 
        %long-run opinions apparently converge to such a consensus. 
        %In 
        %Example \ref{example:probinv3}, where $\mathfrak{D}$ is hard
        %opposition, $\text{Fix}(\mathfrak{D})=\emptyset$ but the
        %corollary does not apply since hard opposition is not a
        %continuous deviation function. 
\end{example}

\subsection*{Polarization}
       We now investigate \emph{(bi-)polarization} as an outcome of our
      opinion updating dynamics. 
      We call an opinion vector $\mathbf{p}\in S^n$
      a \emph{(bi-)polarization} if $\mathbf{p}$ consists of two 
      elements $a,b\in S$ exclusively, that is, if
      $[\mathbf{p}]_i\in\set{a,b}$ for all $i=1,\ldots,n$. Note that
      according to our 
      definition, a consensus vector is a special case of a
      polarization.  
      We first define network structures that are sufficient for
      inducing polarization opinion vectors. 
      \begin{definition}[Opposition bipartite operator
      $\mathbf{W}\odot\mathbf{F}$] 
      We call the operator $\mathbf{W}\odot\mathbf{F}$ 
      %on $n$
      %agents 
      \emph{opposition bipartite} if there exists a
      partition $(\struct{N}_1,\struct{N}_2)$ of the set
      of agents $[n]$ into two disjoint subsets ---
      $[n]=\struct{N}_1\cup\struct{N}_2$, with
      $\struct{N}_1\cap\struct{N}_2=\emptyset$
      --- such that agents in
      $\struct{N}_k$ follow each other, for $k=1,2$, while for all
      agents $i\in \struct{N}_k$, $j\in\struct{N}_{-k}$, for
      $k=1,2$, it holds that $i$ deviates from $j$. More
      precisely, we require 
      \begin{align*}
      \forall\: i,i'\in \struct{N}_k&\Bigl(W_{ii'}>0\implies
      %F_{i_0i_1}=\mathfrak{F}
      i'\in\text{In}(i)
      \Bigr),\quad \text{for } k=1,2,\\
      \forall\:
      i\in\struct{N}_k,j\in\struct{N}_{-k}&\Bigl(W_{ij}>0\implies
      %F_{i_0j_0}\neq \mathfrak{F}
      j\in\text{Out}(i)
      \Bigr),\quad \text{for } k=1,2. 
      \end{align*}
      \end{definition}
      \begin{remark}
      Note that our above definition of opposition bipartiteness is
      equivalent to the condition that (1) no two agents within each
      subsociety are enemies of each other and that (2) no two agents
      across the two subsocieties are friends of each other. 

      Also note that we do not necessarily require $\struct{N}_1$ and
      $\struct{N}_2$ to be non-empty. Therefore, the standard 
      DeGroot model (with exclusively friendship relationships) constitutes a
      special case of an opposition bipartite network in which
      $\struct{N}_1=[n]$ and $\struct{N}_2=\emptyset$. 
      \end{remark}
      \begin{remark}
      What we call `opposition bipartite' operator --- or at least a
      special case of our concept --- has also been
      called `balanced signed network' in the
      literature \parencite[cf.][]{Beasley2010}. 
      %Unrelatedly, 
      %note moreover that, in an 
      %opposite bipartite partitioning $(\struct{N}_1,\struct{N}_2)$,
      %we only take into account links 
      %for which $W_{ij}>0$. To clarify, 
      %%$(\struct{N}_1,\struct{N}_2)$ is a respective partition, then,
      %if, for instance, $i,j\in\struct{N}_1$ and $W_{ij}=0$, we ignore
      %whether $F_{ij}=F$ or $F_{ij}=D$. 
      \end{remark}
      \begin{definition}[Reverse opposition bipartite operator
      $\mathbf{W}\odot\mathbf{F}$] 
      We call the operator $\mathbf{W}\odot\mathbf{F}$ 
      %on $n$
      %agents 
      \emph{reverse opposition bipartite} if there exists a
      partition $(\struct{N}_1,\struct{N}_2)$ of the set
      of agents $[n]$ into two disjoint subsets 
      such that agents in
      $\struct{N}_k$ deviate from each other, for $k=1,2$, while for all
      agents $i\in \struct{N}_k$, $j\in\struct{N}_{-k}$, for
      $k=1,2$, it holds that $i$ follows $j$. More precisely, we
      require 
      \begin{align*}
      \forall\: i,i'\in \struct{N}_k&\Bigl(W_{ii'}>0\implies
      %F_{i_0i_1}\neq\mathfrak{F}
      i'\in\text{Out}(i)
      \Bigr),\quad \text{for } k=1,2,\\
      \forall\:
      i\in\struct{N}_k,j\in\struct{N}_{-k}&\Bigl(W_{ij}>0\implies
      %F_{i_0j_0}= \mathfrak{F}
      j\in\text{In}(i)
      \Bigr),\quad \text{for } k=1,2. 
      \end{align*}
      \end{definition}
      An example of
      an opposition bipartite operator is given in
      Example \ref{example:bip} below. An example of a reverse opposition
      bipartite operator is given in Example \ref{example:opp} below. 
      %and \ref{example:opp2} above. 
      A schematic illustration of both
      concepts is given in Figure \ref{fig:bipOpposition}. 
      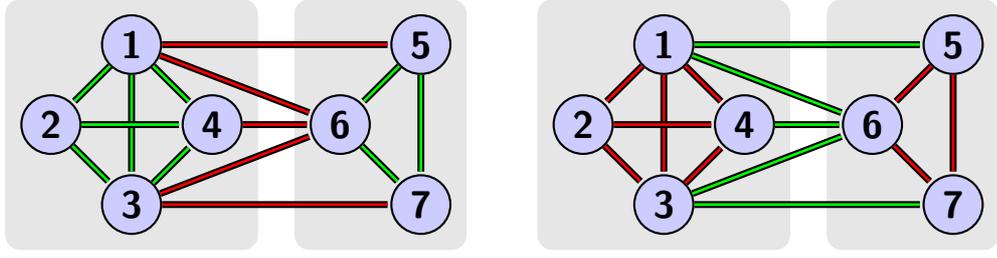
\begin{figure}[!htb]
\centering      
\begin{tikzpicture}[-,>=stealth',shorten >=1pt,auto,node distance=1.5cm,
  thick,main node/.style={circle,fill=blue!20,draw,font=\sffamily\Large\bfseries}]

  \begin{scope}
  \node[main node] (1) {1};
  \node[main node] (2) [below left of=1] {2};
  \node[main node] (3) [below right of=2] {3};
  \node[main node] (4) [below right of=1] {4};

  \node[main node] (5) [right= 3cm of 1] {5};
  \node[main node] (6) [below left of=5] {6};
  \node[main node] (7) [below right of=6] {7};

  \tikzset{Friend/.style   = {
                                 double          = green,
                                 double distance = 1pt}}
  \tikzset{Enemy/.style   = {
                                 double          = red,
                                 double distance = 1pt}}

  \draw[Friend](1) to (2); 
  \draw[Friend](1) to (3);
  \draw[Friend](1) to (4);
  \draw[Friend](2) to (3);
  \draw[Friend](2) to (4);
  \draw[Friend](3) to (4);

  \draw[Friend](5) to (6);
  \draw[Friend](5) to (7);
  \draw[Friend](6) to (7);

  \draw[Enemy](1) to (5);
  \draw[Enemy](1) to (6);
  %\draw[Enemy](1) to (7);
  \draw[Enemy](3) to (6);
  \draw[Enemy](3) to (7);
  \draw[Enemy](4) to (6);
  \begin{pgfonlayer}{background}
    \filldraw [line width=4mm,join=round,black!10]
      (1.north  -| 2.west)  rectangle (3.south  -| 4.east)
      (5.north -| 6.west) rectangle (7.south -| 7.east);
  \end{pgfonlayer}
\end{scope}
  \begin{scope}[xshift=7cm]
  \node[main node] (1) {1};
  \node[main node] (2) [below left of=1] {2};
  \node[main node] (3) [below right of=2] {3};
  \node[main node] (4) [below right of=1] {4};

  \node[main node] (5) [right= 3cm of 1] {5};
  \node[main node] (6) [below left of=5] {6};
  \node[main node] (7) [below right of=6] {7};

  \tikzset{Friend/.style   = {
                                 double          = green,
                                 double distance = 1pt}}
  \tikzset{Enemy/.style   = {
                                 double          = red,
                                 double distance = 1pt}}

  \draw[Enemy](1) to (2); 
  \draw[Enemy](1) to (3);
  \draw[Enemy](1) to (4);
  \draw[Enemy](2) to (3);
  \draw[Enemy](2) to (4);
  \draw[Enemy](3) to (4);

  \draw[Enemy](5) to (6);
  \draw[Enemy](5) to (7);
  \draw[Enemy](6) to (7);

  \draw[Friend](1) to (5);
  \draw[Friend](1) to (6);
  %\draw[Enemy](1) to (7);
  \draw[Friend](3) to (6);
  \draw[Friend](3) to (7);
  \draw[Friend](4) to (6);
  \begin{pgfonlayer}{background}
    \filldraw [line width=4mm,join=round,black!10]
      (1.north  -| 2.west)  rectangle (3.south  -| 4.east)
      (5.north -| 6.west) rectangle (7.south -| 7.east);
  \end{pgfonlayer}
\end{scope}
\end{tikzpicture}
\caption[Schematic illustration of the concepts of opposition
      bipartite and reverse opposition bipartite operators]{Schematic
      illustration of the concepts of opposition 
      bipartite (left) and reverse opposition bipartite operators
      (right). We omit network 
      links referring to weights $\mathbf{W}$ for clarity and we also
      draw links as undirected for the same reason. We omit links
      $F_{ij}$ where $W_{ij}=0$. Red links denote opposition
      ($\mathfrak{D}$), green  
      links following ($\mathfrak{F}$).}
\label{fig:bipOpposition}
\end{figure}

      Clearly, opposition bipartite networks have polarization
      opinion vectors as %fixed-points.
      equilibria, when agents in the same group hold one opinion and
      agents in the alternative group hold the `opposite' opinion, as
      we illustrate below. 
%      \begin{proposition}\label{prop:polarization}
      %In the current setup: If $\mathbf{F}$ is opposition bipartite and if
      %$D(x_0)=x_1$ and $D(x_1)=x_0$ for some $x_0,x_1\in S$, then
      %$(\mathbf{W}\odot\mathbf{F})(\mathbf{p})=\mathbf{p}$, where
      %$\mathbf{p}$ is an opinion vector consisting of entries $x_0$
      %and $x_1$ only.
%      For a social network $\mathbf{W}\odot\mathbf{F}$, let
%      $\mathbf{F}$ be such that
%      $F_{ij}\in\set{\mathfrak{F},\mathfrak{D}}$ for an arbitrary
%      deviation function $\mathfrak{D}$. 
%      Moreover, let $\mathbf{W}\odot\mathbf{F}$ be opposition
%      bipartite and let
%      $a,b\in S$ be $\mathfrak{D}$-opposing viewpoints. Then, there exists a
%      polarization opinion vector $\mathbf{p}$ consisting of
%      opinions $a$ and $b$ such that 
%      $(\mathbf{W}\odot\mathbf{F})\mathbf{p}=\mathbf{p}$. 
%      \end{proposition}
      %Note that both hard and soft opposition, as defined above,
      %satisfy the requirements on $D$ stated in
      %Observation \ref{obs:polarization}. For soft opposition, there
      %are even multiple solutions for the requirement $D(x_0)=x_1$ and
      %$D(x_1)=x_0$. 
      \begin{example}\label{example:bip}
      Let $\mathbf{W}$ be arbitrary row-stochastic. 
      Consider
      \begin{align*}
      \mathbf{F} = \begin{pmatrix}
      \mathfrak{F} & \mathfrak{F} & \mathfrak{D} & \mathfrak{D} \\
      \mathfrak{F} & \mathfrak{F} & \mathfrak{D} & \mathfrak{D} \\
      \mathfrak{D} & \mathfrak{D} & \mathfrak{F} & \mathfrak{F} \\
      \mathfrak{D} & \mathfrak{D} & \mathfrak{F} & \mathfrak{F}
      \end{pmatrix}.
      \end{align*}
      %and let $\mathbf{W}$ be %the uniform weight matrix, with each
      %entry equal to $1/4$. 
      %arbitrary. 
      Clearly, $\mathbf{W}\odot\mathbf{F}$ is opposition bipartite;
      for example, take $\struct{N}_1=\set{1,2}$ and
      $\struct{N}_2=\set{3,4}$. 
      Moreover, let
      $S=\set{\text{``impossible'',``unlikely'',``possible'',``likely'',``certain''}}$
      %as in Section \ref{sec:examples} with 
      and let, e.g., ``unlikely'' and ``likely''
      be \emph{$\mathfrak{D}$-opposing viewpoints}, i.e.,
      $\mathfrak{D}(a)=b$ and $\mathfrak{D}(b)=a$ for $a=$``unlikely''
      and $b=$``likely''. Then
      \begin{align*}
      \mathbf{p}=\begin{pmatrix} \text{``unlikely''}\\ \text{``unlikely''}\\ \text{``likely''}\\ \text{``likely''}\end{pmatrix}
      \end{align*}
      is a polarization fixed-point of $\mathbf{W}\odot\mathbf{F}$.  
      %Example of bipartite and polarization. Also show that it does
      %not need to be strictly bipartite; maybe in a separate example.
      \end{example}
      Reverse opposition bipartite networks may induce
      %, opposition
      %periodic networks do not induce polarization but 
      oscillating, or fluctuating, opinion
      updates \parencite[cf.][]{Kramer1971}, very  
      similar to 
      ordinary periodic 
      networks. 
      %(e.g., %see the standard example in \textcite{Golub2010}).
      %Example \ref{example:peri} below).  
%      \begin{proposition}\label{prop:periodic}
      %In the current setup: If $\mathbf{F}$ is opposition bipartite and if
      %$D(x_0)=x_1$ and $D(x_1)=x_0$ for some $x_0,x_1\in S$, then
      %$(\mathbf{W}\odot\mathbf{F})(\mathbf{p})=\mathbf{p}$, where
      %$\mathbf{p}$ is an opinion vector consisting of entries $x_0$
      %and $x_1$ only.
%      For a social network $\mathbf{W}\odot\mathbf{F}$, let
%      $\mathbf{F}$ be such that
%      $F_{ij}\in\set{\mathfrak{F},\mathfrak{D}}$ for an arbitrary
%      deviation function $\mathfrak{D}$. 
%      Moreover, let $\mathbf{W}\odot\mathbf{F}$ be reverse opposition bipartite and let
%      $a,b\in S$ be opposing viewpoints. Then, there exist
%      polarization opinion vectors $\mathbf{p},\mathbf{\bar{p}}\in S^n$
%      consisting of 
%      opinions $a$ and $b$ in a complementary manner ---
%      $p_i=\mathfrak{D}(\bar{p}_i)$ and $\bar{p}_i=\mathfrak{D}(p_i)$ for all $i=1,\ldots,n$
%      --- 
%      such that 
%      $(\mathbf{W}\odot\mathbf{F})\mathbf{p}=\mathbf{\bar{p}}$ and
%      $(\mathbf{W}\odot\mathbf{F})\mathbf{\bar{p}}=\mathbf{{p}}$. 
%      \end{proposition}
      \begin{example}\label{example:opp}
            Let $\mathbf{W}$ be arbitrary row-stochastic. 
      Consider
      \begin{align*}
      \mathbf{F} = \begin{pmatrix}
      \mathfrak{D} & \mathfrak{D} & \mathfrak{F} & \mathfrak{F} \\
      \mathfrak{D} & \mathfrak{D} & \mathfrak{F} & \mathfrak{F} \\
      \mathfrak{F} & \mathfrak{F} & \mathfrak{D} & \mathfrak{D} \\
      \mathfrak{F} & \mathfrak{F} & \mathfrak{D} & \mathfrak{D}
      \end{pmatrix}.
      \end{align*}
      %and let $\mathbf{W}$ be %the uniform weight matrix, with each
      %entry equal to $1/4$. 
      %arbitrary. 
      Clearly, $\mathbf{W}\odot\mathbf{F}$ is reverse opposition bipartite;
      for example, take $\struct{N}_1=\set{1,2}$ and
      $\struct{N}_2=\set{3,4}$. For $\mathbf{p}$ and $\mathfrak{D}$ as in
      Example \ref{example:bip}, we have
      \begin{align*}
      \begin{pmatrix} 
      \text{``unlikely''}\\ 
      \text{``unlikely''}\\ 
      \text{``likely''}\\ 
      \text{``likely''}
      \end{pmatrix} \mapsto_{\mathbf{W}\odot\mathbf{F}}
      \begin{pmatrix} 
      \text{``likely''}\\ 
      \text{``likely''}\\ 
      \text{``unlikely''}\\ 
      \text{``unlikely''}
      \end{pmatrix} \mapsto_{\mathbf{W}\odot\mathbf{F}} 
      \begin{pmatrix} 
      \text{``unlikely''}\\ 
      \text{``unlikely''}\\ 
      \text{``likely''}\\ 
      \text{``likely''}
      \end{pmatrix}  \mapsto_{\mathbf{W}\odot\mathbf{F}}\ldots
      \end{align*}
      Also, note that in this example self-weights $W_{ii}$ may be zero for all
            agents $i$, so that agents do not necessarily have to
            deviate from their own opinions in order for reverse
            opposition bipartiteness to be satisfied.\footnote{In
            fact, in reverse opposition bipartite networks we have
            either 
            $W_{ii}=0$ or we have $W_{ii}>0$ and $F_{ii}=\mathfrak{D}$, while
            in opposition bipartite networks we have either $W_{ii}=0$
            or we have $W_{ii}>0$ and $F_{ii}=\mathfrak{F}$.}
      \end{example}
      Next, we turn from bi-polarization to \emph{multi-polarization}
      in which an opinion vector $\mathbf{p}\in S^n$ consists of $K$
      distinct opinions $s_1,\ldots,s_K$. We generalize our notion of
      opposition bipartite networks.
      \begin{definition}[Opposition multi-partite operator
      $\mathbf{W}\odot\mathbf{F}$] 
      We call the operator $\mathbf{W}\odot\mathbf{F}$ 
      %on $n$
      %agents 
      \emph{opposition $K$-partite} if there exists a
      partition $(\struct{N}_1,\ldots,\struct{N}_K)$, for $K\ge 2$, of the set
      of agents $[n]$ into disjoint subsets ---
      $[n]=\struct{N}_1\cup\cdots\cup\struct{N}_K$, with
      $\struct{N}_k\cap\struct{N}_{\ell}=\emptyset$ for all $k\neq \ell$ 
      %$\struct{N}_i\neq\emptyset$, for $i=1,\ldots,k$ 
      --- such that agents in
      $\struct{N}_k$ follow each other, for $k=1,\ldots,K$, while for
      all agents $i\in \struct{N}_k$, 
      $j\in\struct{N}_{\ell}$ with $k\neq \ell$, it holds that $i$
      deviates from $j$. More 
      precisely, we require 
      \begin{align*}
      \forall\: i,i'\in \struct{N}_k&\Bigl(W_{ii'}>0\implies
      %F_{i_0i_1}=\mathfrak{F}
      i'\in\text{In}(i)
      \Bigr),\quad \text{for } k=1,\ldots,K,\\
      \forall\:
      i\in\struct{N}_k,j\in\struct{N}_{\ell}&\Bigl(W_{ij}>0\implies
      %F_{i_0j_0}\neq \mathfrak{F}
      j\in\text{Out}(i)
      \Bigr),\quad \text{for } k\neq \ell. 
      \end{align*}
      We also call an opposition $K$-partite
      operator \emph{opposition multi-partite}.
      \end{definition}             
      Opposition multi-partite networks admit
      multi-polarizations as outcomes, as a simple generalization of
      %Proposition \ref{prop:polarization}. 
      %the above example. 
      Example \ref{example:bip}.
        \begin{proposition}\label{prop:multipolarization}
      %For a social network $\mathbf{W}\odot\mathbf{F}$, 
      Let $\mathbf{W}\odot\mathbf{F}$ be an opposition $K$-partite,
      for $K\ge 2$, 
      social network. 
      Moreover, let
      $\mathbf{F}$ be such that
      $F_{ij}\in\set{\mathfrak{F},\mathfrak{D}_1,\ldots,\mathfrak{D}_K}$
      for %an arbitrary 
      deviation functions $\mathfrak{D}_1,\ldots,\mathfrak{D}_K$ that
      precisely correspond to the $K$ groups society is made up of
      (that is, agents in group $k$, for $1\le k\le K$, apply
      deviation function $\mathfrak{D}_k$). 
      %assume that 
      %agents' deviation functions do not vary across out-group agents
      %and that deviation functions $\mathfrak{D}_{\ell}$ correspond to
      %agents in groups $\ell$, for $1\le \ell\le k$.  
      %Moreover, let $\mathbf{W}\odot\mathbf{F}$ be opposition
      %k$-partite and 
      Let
      $s_1,\ldots,s_K\in S$ be such that 
      \begin{align*}
      \mathfrak{D}_k(s_\ell)=s_k,\quad \forall\:k\neq \ell.
      \end{align*}
      %$\mathfrak{D}$-opposing viewpoints. 
      Then, there exists a
      \emph{multi-polarization opinion vector} $\mathbf{p}$ consisting of
      opinions $s_1,\ldots,s_K$, that is,
      $[\mathbf{p}]_i\in\set{s_1,\ldots,s_K}$, such that 
      $(\mathbf{W}\odot\mathbf{F})\mathbf{p}=\mathbf{p}$. 
      \end{proposition}
      %\begin{remark}
      %Apparently, in the last proposition, all
      %$\mathfrak{D}_1,\ldots,\mathfrak{D}_k$ have to be distinct
      %whenever $k>2$. 
      %\end{remark}
      %[MULTI-POLARIZATION - OPPOSIOTN MULTI-PARTITE!; TAKE EXAMPLE 3.2]
      %We conclude with an example of how to induce more general 
      %polarization outcomes, between more than two groups of agents,
      %and a simulation of the discrete weighted majority opinion
      %updating model \eqref{eq:matrix2}. In the latter example, rather
      %than discussing (possible or impossible) fixed points of
      %operators, we simulate \emph{actual} dynamics. 
      \begin{example}\label{example:multiple}
      %We briefly discuss how to induce, in a general manner,
      %polarizing viewpoints between 
      %more than two groups of agents as fixed-points of the operator
      %$\mathbf{W}\odot\mathbf{F}$. %We illustrate with the case of three
      %groups. 
      %One way to achieve such more fragmented opinion and
      %belief systems in society in our setup 
      %is to endow the different groups with \emph{different} deviation
      %functions $D_k:S\goesto S$, where $k$ ranges over the
      %groups (or agents). In essence, these different deviation
      %functions would 
      %represent distinct interpretations of what the opposite of a
      %certain opinion value $a\in S$ is. For example, one group might
      %interpret opposition in a radical manner, allowing $D_k$ to have
      %no fixed-points while other groups may be more `tolerant',
      %opposing only particular opinion values, if held by particular
      %agents. 
      %leaving some opinion values unchanged, even in opposition
      %modus.\footnote{A generalization is to let the deviation
      %endomorphisms depend, not only on the agents who oppose, but
      %also on
      %the opposed agents.}

      For a concrete example, let $S=\set{L,M,R}$ and
      consider three different groups with distinct deviation
      functions $\mathfrak{D}_1(x)=L$, 
      $\mathfrak{D}_2(x)=M$, $\mathfrak{D}_3(x)=R$ for all $x\in \set{L,M,R}$.
      Group
      $1$ may be thought of as always deviating to a left wing opinion, 
      %, at
      %least within the set $\set{A,B,C}$, 
      provided that its members deviate from certain agents; group $2$ to a
      moderate position in the opinion spectrum; and group $3$ to a
      %extreme 
      right wing position. 
      %For illustrative purposes, we want
      %to label the three groups as ``communists'', ``intellectuals'',
      %and ``nazis'', respectively. 
      Let, e.g., $n=6$, $\mathbf{W}$ be arbitrary row-stochastic, and let 
      \begin{align*}
      %\mathbf{W} = 
      %\frac{1}{6}
      %\begin{pmatrix}
      % 1 & 1 & 1 & 1 & 1 & 1\\
      %           1 & 1 & 1 & 1 & 1 & 1\\
      %           1 & 1 & 1 & 1 & 1 & 1\\
      %           1 & 1 & 1 & 1 & 1 & 1\\
      %           1 & 1 & 1 & 1 & 1 & 1\\
      %\end{pmatrix},\quad
      \mathbf{F} = 
      \begin{pmatrix}
      \mathfrak{F} & \mathfrak{F} & \mathfrak{F} & \mathfrak{D}_1 & \mathfrak{D}_1 & \mathfrak{D}_1 \\
      \mathfrak{F} & \mathfrak{F} & \mathfrak{F} & \mathfrak{D}_1 & \mathfrak{D}_1 & \mathfrak{D}_1 \\
      \mathfrak{F} & \mathfrak{F} & \mathfrak{F} & \mathfrak{D}_1 & \mathfrak{D}_1 & \mathfrak{D}_1 \\
      \mathfrak{D}_2 & \mathfrak{D}_2 & \mathfrak{D}_2 & \mathfrak{F} & \mathfrak{D}_2 & \mathfrak{D}_2\\
      \mathfrak{D}_3 & \mathfrak{D}_3 & \mathfrak{D}_3 & \mathfrak{D}_3 & \mathfrak{F} & \mathfrak{F} \\
      \mathfrak{D}_3 & \mathfrak{D}_3 & \mathfrak{D}_3 & \mathfrak{D}_3 & \mathfrak{F} & \mathfrak{F} 
      \end{pmatrix}.
      \end{align*}
      Clearly, $\mathbf{W}\odot\mathbf{F}$ is opposition $3$-partite
      (see also Figure \ref{fig:generalizedPol}). 
      %We then have a partition
      %$(\struct{N}_1,\struct{N}_2,\struct{N}_3)$ of the agent set ---
      %$(\set{1,2,3},\set{4},\set{5,6})$ in the example --- such that
      %agents within each subset $\struct{N}_i$ follow each other and
      %agents across the subsets deviate from each other, applying
      %their specific choices of 
      %deviation functions. 
      It is easy to check that, e.g.,
      $\mathbf{p}=(L,L,L,M,R,R)$ is a fixed-point of
      $\mathbf{W}\odot\mathbf{F}$, in accordance with the proposition.
      %, constituting a `generalized'
      %polarization opinion vector. 
      \begin{figure}[!htb]
      \centering
      \begin{tikzpicture}[-,>=stealth',shorten >=1pt,auto,node distance=1.5cm,
  thick,main node/.style={circle,fill=blue!20,draw,font=\sffamily\Large\bfseries}]

  \begin{scope}
  \node[main node] (1) {1};
  \node[main node] (2) [below left of=1] {2};
  \node[main node] (3) [below right of=1] {3};

  \node[main node] (4) [below of=3] {4};

  \node[main node] (5) [right=2.5cm of 1] {5};
  \node[main node] (6) [below of=5] {6};

  \tikzset{Friend/.style   = {
                                 double          = green,
                                 double distance = 1pt}}
  \tikzset{Enemy/.style   = {
                                 double          = red,
                                 double distance = 1pt}}
  \tikzset{Enemy1/.style   = {
                                 double          = orange,
                                 double distance = 1pt}}
  \tikzset{Enemy2/.style   = {
                                 double          = brown,
                                 double distance = 1pt}}

  \draw[Friend](1) to (2); 
  \draw[Friend](1) to (3);
  \draw[Friend](2) to (3);

  \draw[Friend](4) to [out=-40,in=-140,looseness=0.8,loop,distance=1cm] (4); 

  \draw[Friend](5) to (6);

  \draw[Enemy1,->,dashed](1) to (5);
  \draw[Enemy2,->](6) to (3);
  \draw[Enemy2,->](6) to (4);
  \draw[Enemy1,->,dashed](2) to (4);
  \draw[Enemy,->](4) to (3);
  \draw[Enemy,->](4) to[bend left] (2); %to[out=180,in=-90] (2);
  %\draw[Enemy,->](4) to[bend right] (6); %to[out=0,in=-90] (6);
  \draw[Enemy,->] (4) to [bend right] (6);

  %\begin{pgfonlayer}{background}
  %  \filldraw [line width=4mm,join=round,black!10]
  %    (1.north  -| 2.west)  rectangle (3.south  -| 4.east)
  %    (5.north -| 6.west) rectangle (7.south -| 7.east);
  %\end{pgfonlayer}
\end{scope}
\end{tikzpicture}
      \caption[Graphical illustration of
      Example \ref{example:multiple}]{Graphical illustration of
      Example \ref{example:multiple}. Groups have individualized
      deviation functions, in different colors. 
      We omit many links for clarity.}
      \label{fig:generalizedPol}
      \end{figure}
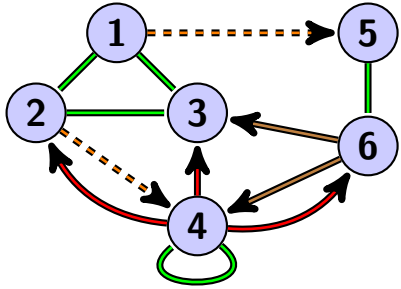
      \end{example}
      %\begin{example}
      %The social network described in
      %Example \ref{example:complexSociety} is opposition
      %$4$-partite. A multi-polarization opinion vector $\mathbf{p}$ for
      %which $(\mathbf{W}\odot\mathbf{F})\mathbf{p}=\mathbf{p}$ is, for
      %example, $\mathbf{p}=(1,1,-1,0,0,0)^\intercal$.\footnote{Example
      %\ref{example:complexSociety} is actually not covered by
      %Proposition \ref{prop:multipolarization} since the deviation
      %functions of agent $6$ vary across out-group agents.} 
      %%But there may be
      %%many others, not covered by
      %%Proposition \ref{prop:multipolarization}. Figure \ref{fig:}
      %\end{example}
      We omit the introduction of reverse opposition $K$-partite
      networks, for $K\ge 2$, as a straightforward
      generalization of reverse opposition bipartite networks. The
      generalization is along the
      lines of the generalization of opposition $K$-partite networks
      over opposition bipartite ones.

\section{Continuous linear deviation functions}\label{sec:distrust}
We now consider the situation when deviation functions take the
affine-linear form 
\begin{align}\label{eq:affine_linear}
  \mathfrak{D}(x) = a\cdot x+b,
\end{align}
for $a,b\in \real$. When deviation functions are of this form and when
agents are  
in addition homogeneous with respect to their deviation
functions, i.e., $\mathfrak{D}$ does not vary across agents, then the
operator $\mathbf{W}\odot\mathbf{F}$ admits a particularly simple form,
namely, that of an affine-linear operator; see
Proposition \ref{prop:consensusOpposition} in the appendix. Moreover,
we consider here the  
special case when $a=-1$ and $b=0$.\footnote{The case $|a|>1$ usually
implies that opinions `explode' over time because agents `overscale' their
out-group members' opinions while $|a|<1$ usually implies that society
reaches a `zero consensus' because agents iteratively discount their
out-group members' opinions.} 
We consider such $\mathfrak{D}$ on opinion spectra $S_0$ 
that are either of the form $[-\beta,\beta]$ for some $\beta>0$ or of the form 
$S=\real$. We call this $\mathfrak{D}$ defined on such opinion spectra
\emph{soft opposition}.\footnote{We call such $\mathfrak{D}$ soft
  opposition because opinions are inverted the less strongly the
  closer they are to the neutral consensus zero. In our working paper
  version, we 
also define a (discontinuous) deviation function called \emph{hard
  opposition} that maps opinions to extreme inverted values of the opinion
spectrum, depending on whether they are positive or negative
(depending on which side of the opinion spectrum they lie, with $0$ as
reference point).} 
% as introduced
%before. 
In this case,
social networks $\mathbf{W}\odot\mathbf{F}$ may be represented by a
matrix $\mathbf{A}$ that has entries $W_{ij}$ iff $j\in\text{In}(i)$
and entries $-W_{ij}$ iff $j\in\text{Out}(i)$. 

\subsection{A graph theoretical description}
In the remainder, we consider the situation when graphs
(i) contain no self-loops and are undirected (simple) in the sense
that $W_{ij}=W_{ji}$ and $F_{ij}=F_{ji}$,\footnote{This captures
  \emph{reciprocity}: amity and enmity are mutual.}  
(ii) contain at most one
type of deviation function $\mathfrak{D}$ across all agents, and (iii)
when
$\mathfrak{D}$ is linear on the opinion spectrum $S_0$; in particular,
$\mathfrak{D}$ is soft opposition. 
When networks $\mathbf{W}\odot\mathbf{F}$ are so specified, then, as
before, $\mathbf{W}\odot\mathbf{F}$ admits a %affine-
linear matrix 
representation %of the form %$(\mathbf{A},\mathbf{d})$ 
$\mathbf{A}$ where 
%and 
$\mathbf{A}$ 
is in addition symmetric: $\mathbf{A}^\intercal=\mathbf{A}$. 
We denote this
class of networks by $\SSLS(S_0)$ (for 
\emph{S}imple, \emph{L}inear, \emph{S}oft opposition, and where the argument
refers to the opinion spectrum), that is, 
%Formally, we have the
%following definition.  
%\begin{definition}
  %Let $S$ be an opinion spectrum of the form $S=[\alpha,\beta]$ or
  %$S=\real$. Then, 
  %Let $\mathbf{W}\odot\mathbf{F}$ be a social network with agent set
  %$[n]$. 
%  Let 
  \begin{align*}
    \SSLS(S_0) = \set{\mathbf{W}\odot\mathbf{F}\sd \forall
      i,j\in[n]\:
      \bigl( &W_{ii}=0,\:{F}_{ij}\in\set{\mathfrak{F},\mathfrak{D}},\:
      \mathfrak{D} \text{
      soft opposition on
      $S_0$},\\
      & W_{ij}=W_{ji}, F_{ij}=F_{ji}\bigr)}  
  \end{align*}
  denotes the class of social networks on agent sets $[n]$ that satisfy
  simplicity, 
  symmetricity, etc., as described.
  %discussed above. Here $\mathbf{A}$ is the
  %linear representation of $\mathbf{W}\odot\mathbf{F}$ as outlined. 
  %denote the class of simple graphs cont
  %Symmetric linear, soft opposition
%\end{definition}
  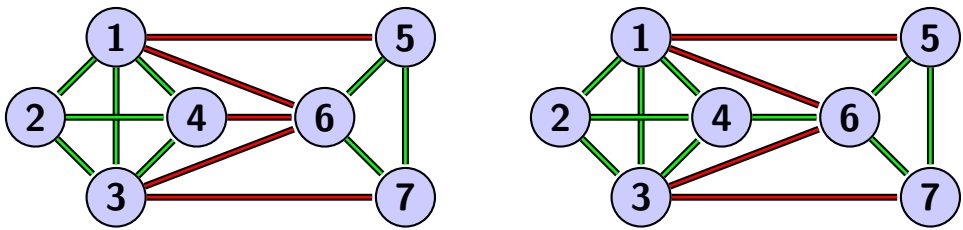
\begin{figure}[!htb]
    \centering
    % use -> to make arrows
\begin{tikzpicture}[-,>=stealth',shorten >=1pt,auto,node distance=1.5cm,
  thick,main node/.style={circle,fill=blue!20,draw,font=\sffamily\Large\bfseries}]

  \node[main node] (1) {1};
  \node[main node] (2) [below left of=1] {2};
  \node[main node] (3) [below right of=2] {3};
  \node[main node] (4) [below right of=1] {4};

  \node[main node] (5) [right= 3cm of 1] {5};
  \node[main node] (6) [below left of=5] {6};
  \node[main node] (7) [below right of=6] {7};

  \tikzset{Friend/.style   = {
                                 double          = green,
                                 double distance = 1pt}}
  \tikzset{Enemy/.style   = {
                                 double          = red,
                                 double distance = 1pt}}

  \draw[Friend](1) to (2); 
  \draw[Friend](1) to (3);
  \draw[Friend](1) to (4);
  \draw[Friend](2) to (3);
  \draw[Friend](2) to (4);
  \draw[Friend](3) to (4);

  \draw[Friend](5) to (6);
  \draw[Friend](5) to (7);
  \draw[Friend](6) to (7);

  \draw[Enemy](1) to (5);
  \draw[Enemy](1) to (6);
  %\draw[Enemy](1) to (7);
  \draw[Enemy](3) to (6);
  \draw[Enemy](3) to (7);
  \draw[Enemy](4) to (6);
  %\path[every node/.style={font=\sffamily\small}]
  %  (1) edge node [left] {0.6} (4)
  %      edge [bend right] node[left] {0.3} (2)
  %      %edge [loop above] node {0.1} (1)
  %  (2) edge node [right] {0.4} (1)
  %      edge node {0.3} (4)
  %      %edge [loop left] node {0.4} (2)
  %      edge [bend right] node[left] {0.1} (3)
  %  (3) edge node [right] {0.8} (2)
  %      edge [bend right] node[right] {0.2} (4)
  %  (4) edge node [left] {0.2} (3)
  %      %edge [loop right] node {0.6} (4)
  %      edge [bend right] node[right] {0.2} (1);
\end{tikzpicture}
\hspace{1cm}
\begin{tikzpicture}[-,>=stealth',shorten >=1pt,auto,node distance=1.5cm,
  thick,main node/.style={circle,fill=blue!20,draw,font=\sffamily\Large\bfseries}]

  \node[main node] (1) {1};
  \node[main node] (2) [below left of=1] {2};
  \node[main node] (3) [below right of=2] {3};
  \node[main node] (4) [below right of=1] {4};

  \node[main node] (5) [right= 3cm of 1] {5};
  \node[main node] (6) [below left of=5] {6};
  \node[main node] (7) [below right of=6] {7};

  \tikzset{Friend/.style   = {
                                 double          = green,
                                 double distance = 1pt}}
  \tikzset{Enemy/.style   = {
                                 double          = red,
                                 double distance = 1pt}}

  \draw[Friend](1) to (2); 
  \draw[Friend](1) to (3);
  \draw[Friend](1) to (4);
  \draw[Friend](2) to (3);
  \draw[Friend](2) to (4);
  \draw[Friend](3) to (4);

  \draw[Friend](5) to (6);
  \draw[Friend](5) to (7);
  \draw[Friend](6) to (7);

  \draw[Enemy](1) to (5);
  \draw[Enemy](1) to (6);
  %\draw[Enemy](1) to (7);
  \draw[Enemy](3) to (6);
  \draw[Enemy](3) to (7);
  \draw[Friend](4) to (6);
  %\path[every node/.style={font=\sffamily\small}]
  %  (1) edge node [left] {0.6} (4)
  %      edge [bend right] node[left] {0.3} (2)
  %      %edge [loop above] node {0.1} (1)
  %  (2) edge node [right] {0.4} (1)
  %      edge node {0.3} (4)
  %      %edge [loop left] node {0.4} (2)
  %      edge [bend right] node[left] {0.1} (3)
  %  (3) edge node [right] {0.8} (2)
  %      edge [bend right] node[right] {0.2} (4)
  %  (4) edge node [left] {0.2} (3)
  %      %edge [loop right] node {0.6} (4)
  %      edge [bend right] node[right] {0.2} (1);
\end{tikzpicture}
    \caption[Balanced and unbalanced networks]{Balanced and imbalanced
      networks. The left network is 
      opposition bipartite (balanced) while the right is neither
      opposition bipartite nor reverse opposition bipartite. In
      particular, agents $3$ and $4$ have mutual `friends' (agents
      $1$, $2$) while agent $6$ is in $3$'s out-group and
      $4$'s in-group. Alternatively: agents $4$ and $6$, from
    different positively linked factions, have befriended each other.}
    \label{fig:unbalanced}
  \end{figure}

We make an additional technical assumption here and in the remainder,
namely, %Moreover,  
we
generally assume that the social networks wherein agents interact are
\emph{aperiodic}, that is, the greatest
  common divisor of the lengths of their simple cycles is $1$.
Our main theorem in this context %already 
exhaustively categorizes
long-run opinions in terms of %the 
three different social network
structures as outlined in the theorem. 
%(conditions (i) and (iii) of the
%theorem are exemplified in Figure \ref{fig:unbalanced}). 
In the
theorem, recall that a graph 
$G$ %$G=(V,E)$ 
is said to be 
\emph{(strongly) connected} if there
  exists a path in $G$ from any node to any other node. 
Moreover, in the theorem, 
%we say that $\mathbf{W}\odot \mathbf{F}$ is
%  \emph{convergent for opinion vector $\mathbf{b}(0)\in S^n$} if 
%  $\lim_{t\goesto\infty}(\mathbf{W}\odot\mathbf{F})^t\mathbf{b}(0)$
%  exists. 
%  We say that  
%  $\mathbf{W}\odot\mathbf{F}$ \emph{induces a consensus for opinion
%  vector $\mathbf{b}(0)$} if
%  $\mathbf{W}\odot\mathbf{F}$ is {convergent for $\mathbf{b}(0)$} and
%  $\lim_{t\goesto\infty}(\mathbf{W}\odot\mathbf{F})^t\mathbf{b}(0)$ is a
%  {consensus}.
%  %, that is, a vector $\mathbf{c}\in S^n$ with all
%  %entries identical. % 
  we say that $\mathbf{W}\odot \mathbf{F}$ is
  \emph{convergent} if $\mathbf{W}\odot \mathbf{F}$ is convergent
  for {all} initial opinion vectors $\mathbf{b}(0)\in S^n$, that
  is,
  $\lim_{t\goesto\infty}(\mathbf{W}\odot\mathbf{F})^t\mathbf{b}(0)$
  exists for all $\mathbf{b}(0)\in S^n$. 
  We say that  
  $\mathbf{W}\odot\mathbf{F}$ \emph{induces a consensus} if
  $\mathbf{W}\odot\mathbf{F}$ induces a consensus for {all} initial
  opinion vectors $\mathbf{b}(0)\in S^n$, that is,
  $\mathbf{W}\odot\mathbf{F}$ is convergent and
  $\lim_{t\goesto\infty}(\mathbf{W}\odot\mathbf{F})^t\mathbf{b}(0)$ is
  always a {consensus}. 
  We say that $\mathbf{W}\odot \mathbf{F}$ is
  \emph{divergent} if %$\mathbf{W}\odot\mathbf{F}$ 
  it 
  is not
  convergent. In other words,   $\mathbf{W}\odot \mathbf{F}$ is
  divergent 
  when there exists a vector $\mathbf{b}(0)$ such that
  $\lim_{t\goesto\infty}(\mathbf{W}\odot \mathbf{F})^t\mathbf{b}(0)$
  does not exist. We say that $\mathbf{W}\odot\mathbf{F}$ \emph{induces a
      polarization} if
      $\lim_{t\goesto\infty}(\mathbf{W}\odot\mathbf{F})^t\mathbf{b}(0)$
      is a polarization, for all initial opinion vectors
      $\mathbf{b}(0)$. We say that $\mathbf{W}\odot\mathbf{F}$ induces
      a \emph{non-zero polarization} when $\mathbf{W}\odot\mathbf{F}$ induces
      a polarization and
      $\lim_{t\goesto\infty}(\mathbf{W}\odot\mathbf{F})^t\mathbf{b}(0)\neq
      \mathbf{0}$ for some initial opinion vectors
      $\mathbf{b}(0)$. 
\begin{theorem}\label{theorem:main}
  Let $\mathbf{W}\odot\mathbf{F}\in \SSLS(S_0)$.
  Assume that $\mathbf{W}\odot\mathbf{F}$ is strongly connected (since
  $\mathbf{A}$ is symmetric, we might also simply say `connected') and
  aperiodic. 
  Then: 
    \begin{itemize}
      \item[(i)] $\mathbf{W}\odot\mathbf{F}$ induces a %(non-trivial) 
  polarization (that is not always zero) if and only if
  $\mathbf{W}\odot\mathbf{F}$ is opposition 
  bipartite. 
  \item[(ii)] $\mathbf{W}\odot\mathbf{F}$ is divergent if and only if
    $\mathbf{W}\odot\mathbf{F}$ is reverse opposition 
    bipartite.  
  \item[(iii)] $\mathbf{W}\odot\mathbf{F}$ induces a
  neutral\footnote{Here, we call a consensus $(c,\ldots,c)$ 
    neutral
    if $\mathfrak{D}(c)=c$. For soft opposition on $S_0$, $c=0$.} consensus if
  and only if $\mathbf{W}\odot\mathbf{F}$ is 
  neither opposition
  bipartite nor reverse opposition bipartite. 
    \end{itemize}
  %Let the graph underlying $\mathbf{W}\odot\mathbf{F}$ be strongly
  %connected and aperiodic. Then, 
\end{theorem}
For better understanding, we give alternative characterizations of
conditions (i) and (ii) of Theorem \ref{theorem:main}, which 
follow from the theorem and its proof. Namely, we find that:  
\begin{itemize}
  \item[(i)]
  If and only if $\mathbf{W}\odot\mathbf{F}$ is opposition bipartite,
  the following holds: 
  $\mathbf{b}(\infty)=\lim_{t\goesto\infty}(\mathbf{W}\odot\mathbf{F})\mathbf{b}(0)$
  exists for all $\mathbf{b}(0)\in S^n$ and $\mathbf{b}(\infty)$
  is then always either the zero vector (e.g., when $\mathbf{b}(0)$ is
  the zero vector; but it is not the zero vector for all
  $\mathbf{b}(0)\in S^n$) or is some polarization vector where each entry is
  $a>0$ or $-a$. Moreover, when $W_{i,\text{Out}(i)}>0$ for some
  $i\in[n]$, then, if $\mathbf{b}(\infty)$ is a non-zero polarization vector,
  both $a$ and $-a$ are components of $\mathbf{b}(0)$; otherwise
  $\mathbf{b}(\infty)$ is always a consensus, as in the standard
  DeGroot model.
  \item[(ii)] If and only if $\mathbf{W}\odot\mathbf{F}$ is reverse opposition
  bipartite,  
  the following holds:
  $\lim_{t\goesto\infty}(\mathbf{W}\odot\mathbf{F})\mathbf{b}(0)$ does
  not exist for all $\mathbf{b}(0)$ in $S^n$. Moreover, when it exists
  it is the zero vector. 
  %\item[(iii)] 
\end{itemize}

%\begin{remark}
The fact that polarization requires `exact' balance (opposition
bipartiteness) and admits not a `grain of imbalancedness', as stated in 
Theorem \ref{theorem:main} and exemplified in Figure
\ref{fig:unbalanced}, may appear odd since one might expect, in 
reality, small perturbations to balance --- e.g., small-scale intra-group
antagonisms or individual friendships among enemies --- to be the rule 
rather than the exception, particularly in 
large enough systems.
%\footnote{\textcite{Facchetti2011} empirically
%  demonstrate, 
%  however, that currently available on-line social
%  networks are indeed `extremely balanced'.} 
We note that this result
%is, however, to a large
%may be interpreted as  
%part, due to the
is closely connected to the  
continuous opinion spectrum and the corresponding averaging %DeGroot
updating that we have considered in this section. 
%Should it, in
%contrast,  %the reader thinks
%that 
%be more appropriate to perceive of 
If one thinks that reality is better perceived of as 
%reality is better perceived of as 
%being 
discrete, with weighted majority voting as a more
plausible 
opinion updating mechanism, then %one can be reassured that %, clearly,  
%we note that, as we have
%shown, 
it is apparent that 
the 
discrete model \emph{is} clearly robust against small such perturbations, 
so that polarizing viewpoints can be Nash equilibria in this case
even when the underlying networks exhibit (marginal)
imbalancedness. 
%In addition, we note that our analysis thus far has
%also depended on the specification of weight sum requirements, as we
%illustrate in the following. 
%\end{remark}

We also note that our results may be generalized to \emph{periodic}
graphs (those that are not aperiodic) and to graphs that are not
connected. We leave this to future work. 

\subsection{Social influence and opinion
  leadership}\label{section:influence} 

In DeGroot learning, one of the important questions has been that of
\emph{opinion leadership}: whose agents' initial opinions have most
impact upon resulting limiting (long-run) opinions and how does this
depend on the 
network structure in which the agents are embedded. In the context of
connected 
$\SSLS(S_0)$ networks as we have defined above, this question admits an
elegant solution in our extended DeGroot model with in-group/out-group
relationships. Namely, if the network is reverse opposition bipartite,
the updating operator diverges (for at least some initial opinion
vectors) and opinion leadership is thus not 
well-defined. If, in contrast, the network is opposition bipartite and
aperiodic, opinion leadership is determined by \emph{eigenvector
  centrality} \parencite{Bonacich1972} exactly in the same way as in the
original DeGroot model, except that a plus or minus sign indicates the
group membership of the agents. Finally, if none of those two
conditions hold, then no agent is influential, since agents will
converge to a fixed-point of the deviation function no matter their
initial opinions. %Our next theorem summarizes these observations.
The opposition bipartite case is the focus of the next theorem. 

\begin{theorem}\label{theorem:social_influence}
  Let $\mathbf{W}\odot\mathbf{F}\in \SSLS(S_0)$.
  %, where
  %$S=[-\beta,\beta]$, for some 
  %$\beta>0$, or $S=\real$.  
  Assume that $\mathbf{W}\odot\mathbf{F}$ is strongly
  connected and aperiodic. 
  Then the following are equivalent:
  %\begin{itemize}
  %  \item[(i)] 
      %Each agent's power $v_i$ is given by $v_i=\pm w_i$
      %where $\mathbf{w}$ is the unique left unit eigenvector of
      %$\length{\mathbf{A}}$ with normalization $\sum_{i\in[n]}w_i=1$
  \begin{itemize}
    \item[(i)]
  There exists a unique (nonnegative) left unit
      eigenvector $\mathbf{s}$ of ${\length{\mathbf{A}}}$, the matrix
      with entries $|A_{ij}|$, whose
      entries sum to $1$ such that 
      each agent holds one of two %(not necessarily distinct) 
      long-run 
      opinion values 
      $a$ and $b$ ($b=-a$), given by
      \begin{align*}
        a & = \sum_{j\in[n]} g(j)s_jb_j(0),\\
        b & = \sum_{j\in[n]} (-g(j))s_jb_j(0),
      \end{align*}
      where $g(j)\in\set{\pm 1}$.  
      %(depending on agent $j$'s
      %group membership) 
      %and where $\mathbf{s}$ is the unique left unit
      %eigenvector of ${\length{\mathbf{A}}}$
      %if
      %and only if 
    \item[(ii)] $\mathbf{W}\odot\mathbf{F}$ is opposition
      bipartite. % and aperiodic.
  \end{itemize}
    %\item[(ii)] Opinion leadership is not (well-)defined if and only
    %  if $\mathbf{W}\odot\mathbf{F}$ is reverse opposition
    %  bipartite. 
    %\item[(iii)] Each agent's power is given by $v_i=0$ 
    %  if and only if it is neither opposition
    %  bipartite nor reverse opposition bipartite.
  %\end{itemize} 
\end{theorem} 
Since in case of opposition bipartite and aperiodic strongly connected
networks there are two long-run opinion values and an agent $j$'s
`influence' on each of them is $g(j)s_j$ and $-g(j)s_j$, respectively, where
$g(j)\in\set{\pm 1}$, we may speak of agent $j$'s \emph{absolute
  power} $\length{v_j} = \length{g(j)s_j}=\length{-g(j)s_j}=s_j$.  
Thus, in summary, 
the last theorem and our previous discussion lead to the following
characterization for strongly 
connected networks $\mathbf{W}\odot\mathbf{F}\in \SSLS(S_0)$:
%, as
%indicated:
 \begin{itemize}
    \item[(i)] 
      Each agent's \emph{absolute power} %, his power in absolute sense, 
      $\length{v_i}$ is given by $\length{v_i}=s_i$
      where $\mathbf{s}$ is the unique left unit eigenvector of
      $\length{\mathbf{A}}$ with normalization $\sum_{i\in[n]}s_i=1$
      if and only if $\mathbf{W}\odot\mathbf{F}$ is opposition
      bipartite.
    \item[(ii)] Opinion leadership is not (well-)defined if and only
      if $\mathbf{W}\odot\mathbf{F}$ is reverse opposition
      bipartite. 
    \item[(iii)] Each agent's power is given by $v_i=0$ 
      if and only if $\mathbf{W}\odot\mathbf{F}$ is neither opposition
      bipartite nor reverse opposition bipartite.
  \end{itemize} 

It is noteworthy that in case (i), absolute power is independent of
the kinds of relationships between agents and only depends on their
intensities. In other words, %also strong negative links may increase
an agent may also be prominent when she attracts strong negative links.
%the absolute power of an agent.   
When society partitions into several subsocities, %as in
%Definition \ref{def:structuring}, 
these
results may be applied independently to each of them. 
%Note that the
%rest of the world's initial opinions cannot influence the long-run
%opinions of anyone 
%including themselves, hence they always have power $v_i=0$. 
%always has power $v_i=0$ since the agents that it is
%made up of have no influence upon limiting opinions. 

%\section{Wisdom}\label{sec:wisdom}
%\input{wisdom}

\section{Concluding remarks}\label{sec:conclusions}
 Opinions are important in an economic context (and other contexts)
since they shape the demand for products, set the political
course, and guide, in general, socio-economic behavior. 
%In socio-economics, 
Models of opinion %formation and
dynamics describe how 
individuals form opinions or beliefs about an underlying state or a
discussion topic.
Typically, in the social networks literature,
subjects may 
communicate with other individuals, their peers, in
this context, enabling them to %possibly
aggregate dispersed information. Bayesian models of opinion formation
assume that agents 
form their opinions in a fully rational manner and have an accurate
`model of the world' at their disposal, both of which are questionable
and unrealistic assumptions for %compared to 
actual social learning processes of
human 
individuals. %\parencite[cf.][etc.]{Chandra2012,Corazzini2012}. 
 Non-Bayesian
models, and    
most prominently the classical DeGroot model of opinion formation, 
while also not unproblematic \parencite[cf.][]{Acemoglu2011}, posit that
agents  
employ simple `rule-of-thumb' heuristics to integrate the opinions of
others. Unfortunately, both the non-Bayesian and Bayesian paradigms
typically lead individuals to a 
consensus, which apparently contradicts the facts as people 
disagree with others on many issues of (everyday) life. In the context
of DeGroot learning models, some approaches %have sought to address
                                %this
can address this, 
%issue, 
either by assuming a homophily principle whereby agents limit
their communication to those who hold similar opinions as themselves
or by introducing stubborn agents, modeling, e.g., opinion leaders,
who never update their opinions. Both approaches are, again, 
debatable since %the approach based on stubborn agents assumes truly
%autark individuals and models based on homophily can typically neither
%explain short-term opinion fluctuations \parencite[see the discussion
%in][]{Acemoglu:2013},  
%nor functional disagreement whereby disagreeing opinions are, in fact,
%opposing viewpoints rather than arbitrary and unrelated. 
they assume a complete lack of flow of information between some classes of
agents (from some 
time onwards). 
%Finally, the
%homophily models that can account for disagreement %also 
%rely on
%the condition that some subsets of society do not communicate with, or
%learn from, each other, at least from some time point onward, as in
%the model based on stubborn agents --- a requirement that we %generally
%find problematic since it is difficult to imagine subsets of society
%without \emph{any} mutual 
%influence.\footnote{Particularly in today's `globalized
%world'.} In any case, 
In addition, 
models based on homophily and stubborn
agents both ignore \emph{negative} relationships between individuals
as potential sources for conflict and disagreement.

In the current work, we have 
investigated opinion dynamics under
%opposition, 
%in-group bias and 
out-group discrimination (in-group bias)
as such a potentially alternative explanation for
disagreement. In our 
setup, agents are driven by two forces: they want 
to adjust their opinions to match those of the agents of 
%they follow
%(
their in-group %, friends, or those they trust) 
and, in addition, they want to
adjust their opinions to match the `inverse' of those 
of the agents of 
%they oppose (
their out-group. 
%, enemies, or those they distrust). 
Best responses
in this setting lead us to a DeGroot-like opinion updating process 
%--- 
%which is the classical non-Bayesian learning model for opinion
%dynamics in the (socio-)economics, and related, literature --- 
%whereby
in which 
agents form their next period opinions via weighted %arithmetic
averages of their neighbors' (possibly inverted) opinion signals. 
%Our
%paradigm can account for a variety of phenomena such as consensus,
%neutrality, disagreement, and (functional) polarization, depending upon network
%(multigraph) 
%structures and specifications of deviation functions, as we have
%demonstrated, both analytically and by means of simple
%simulations. Psychologically and socio-economically, we
%have interpreted opposition as arising either from rebels;
%countercultures; rejection of the norms and values of disliked others,
%as `negative referents'; or, simply, distrust. 
Unlike in the standard DeGroot model where opinions may converge to
arbitrary consensus opinion profiles, in our model only neutral
consensus opinion profiles may be attained in the long-run, that is,
consensus vectors where the consensus opinion is a fixed-point of each
agent's deviation function (modeling the mode of opposition
between agents). Thus, if there exists no opinion that is `globally' neutral in
this sense, our model predicts persistent disagreement provided that
negative 
relations between agents are sufficiently strong. When we specialize our
model to undirected networks containing no self-loops and where the only
allowable deviation function is `soft opposition', we derive necessary
and sufficient conditions for bi-polarization in connected
societies. These say that long-run opinions bi-polarize if and only if the
underlying network wherein agents communicate satisfies `opposition
bipartiteness': it consists of two groups of agents exhibiting 
positive within-group links and negative between-group links. 
%negatively interlinked groups of
%agents satisfying positive links. 
We also investigate social influence in this
special case, that is, the question of whose initial opinions matter
most for resulting long-run opinions. We find that in opposition
bipartite networks (satisfying aperiodicity), opinion leadership, in
terms of \emph{absolute power}, is
determined by eigenvector centrality exactly in the same way as in the
standard DeGroot model.
%, except that a plus or minus sign indicates
%group membership. 
This means that an agent is prominent to the degree
that she is interlinked with prominent agents; it is noteworthy that
%hence 
even
strong negative ties increase prominence. 

Finally,
%we have 
considering the question of \emph{wisdom} (\cite{Golub2010,Jadbabaie2012,Buechel2015}), that is, whether agents can
(jointly) learn the true state of nature of their discussion topic,
provided that such a truth exists, we observe the following. 
The case for wisdom is a
weak one in our model since negative ties typically prevent consensus
formation, so clearly not everybody can be wise in the
long-run. 
 This holds \emph{even when agents are
initially perfectly informed} in the sense that their initial beliefs
coincide with truth. In
particular, if agents 
(multi-)polarize, then at most one group of agents may be wise in the
long-run, but due to the mutual dependence of long-run beliefs, we might
suspect none to be.\footnote{In particular, observe that an agent's
wisdom is expected to decrease in her amount of in-group bias (i.e., how
strongly she discriminates against out-group individuals) provided
that out-group individuals' opinions are close to truth, because this
controls how %she
%then 
strongly she desires to match the `opposite' of truth. From a reverse
perspective, her wisdom is expected to increase in her out-groups' biases.} 
Ultimately, this result must be interpreted by reference to 
the rationality of the agents involved. The standard interpretation
of DeGroot learners is that of na\"ive individuals susceptible to
persuasion bias. \textcite{Golub2010} show that such agents can
still learn the true state of nature under not too demanding
conditions. \textcite{Jadbabaie2012} show that slightly more
rationality increases the case for wisdom. In contrast, we show that
an additional bias such as in-group bias may significantly worsen this
case.  

%One issue that has been left undiscussed so far is the fact that,
%possibly unlike %in the area of 
%social norms and values, opinions oftentimes
%(though probably far less than always) admit a `truth' against which
%they may be evaluated; accordingly, some research papers
%\parencite[e.g.,][]{Golub2010} have asked for the conditions under which 
%agents may 
%converge to a consensus that is even correct. Under opposition, as we
%have specified, such a convergence to a correct consensus is severely
%compromised, as we have indicated. Namely, if agents converge to a
%consensus at all, then, as seen, such a consensus is typically a
%neutral consensus. In the continuous approach, if the opinion spectrum
%is a dense subset of the real line and the set of neutral opinions is,
%as we might plausibly assume, small (e.g., finite or even a 
%singleton), then, from a probabilistic perspective, chances for agents
%of reaching a correct consensus are virtually zero. Alternatively, if
%agents disagree, or, more specifically, polarize, then, of course, at
%most one group of agents can be correct but, given a functional
%dependence of limiting opinions, we would expect none to be.

Concerning future research directions within our context,
both weight 
links and opposition links between 
agents, $\mathbf{W}$ and $\mathbf{F}$, have been assumed exogenous in
the current work. 
%In
%future work, 
Prospectively, 
it might be worthwhile to consider endogenous link
formation processes. In particular, the origin and evolution of
%opposition behavior, 
out-group relations, 
and %its relation to
their %connection to
interdependence with    
agents' opinions and external
factors, such as, most importantly, %possibly
external truth, might be of interest.

\section*{Acknowledgments}
I am indebted to the associate editor Juan D.~Moreno-Ternero
and two anonymous reviewers for various comments that have greatly
improved the layout of the current work. 
I also thank Matthias Blonski for valuable suggestions that led
to improvements in the paper. All remaining errors are my
own.

\begin{appendices}
\section{Definitions, theorems, and proofs}\label{sec:appendixChap1}
Sections \ref{sec:graphs} and \ref{sec:matrix} review notation and concepts for
\emph{graphs} and \emph{matrices}, respectively. Section
\ref{sec:signed} states basic results on signed social networks from
\textcite{Altafini2013}. Finally, Section \ref{sec:proofs} provides
the proofs of our own results in this work. 
\subsection{Graphs}\label{sec:graphs}
%Next, we %note that %$n\times n$ 
%%matrices and \emph{networks} are
%%closely related. We first give a formal definition of a network.
%formally introduce \emph{networks} because of their relationship to
%our `matrix' operators $\mathbf{W}\odot\mathbf{F}$.
\begin{definition}[(Weighted) Network]
  A \emph{network}, or \emph{graph}, is a tuple $G=(V,E)$ where $V$
  is a finite set and
  $E\subseteq V\times V=\set{(u,v)\sd u,v\in V}$. We call $V$ the
  \emph{vertices} or \emph{nodes} of graph $G$ and $E$ the
  \emph{edges} or \emph{links} of $G$. 
  Moreover, we call the network $G$ \emph{weighted} if there
  exist \emph{weights} 
  $w_{uv}$ for each edge $(u,v)\in E$.\footnote{Weights may typically be real
  numbers but we more generally allow them to be arbitrary
  mathematical objects.} 
\end{definition}
%We note that the 
%Edges of a network $G$ represent a relationship
%between nodes (agents, in our case), namely, whether or not two nodes
%are connected; weights generalize this binary relationship.
%by specifying its `degree'. %of relationship 
%as an extension of the
%binary edge relation. In
In 
a \emph{multigraph}, instead of having only one link type between
nodes, there may exist multiple link types. The networks
investigated in this work may
be considered multigraphs with exactly two types of links, one
denoting intensity of connection and one denoting kind of connection. 
\begin{definition}
  A \emph{walk} in a network $G=(V,E)$ is a sequence of nodes
  $i_1,i_2,\dotsc,i_K$, not necessarily distinct, such that
  $(i_k,i_{k+1})\in E$ for all 
  $k\in\set{1,\dotsc,K-1}$. 
  A \emph{path} is a walk consisting of distinct nodes. 
  A \emph{cycle} is a walk $i_1,\dotsc,i_K$ such that $i_1=i_K$. The
  \emph{length} of cylce $i_1,\dotsc,i_K$ is defined to be $K-1$. A
  cycle is called \emph{simple} if the only node appearing twice is
  $i_1=i_K$. 
\end{definition} 
%\begin{definition}
%  The graph $G=(V,E)$ is said to be \emph{strongly connected} if there
%  exists a path in $G$ from any node to any other node. 
%\end{definition}
%\begin{definition}
%  The graph $G=(V,E)$ is said to be \emph{aperiodic} if the greatest
%  common divisor of the lengths of its simple cycles is $1$. We call
%  $G$ \emph{periodic} if it is not aperiodic. 
%\end{definition}
\begin{remark}
%We remark that 
We %generally 
use the same terminology --- `strongly
connected', `aperiodic', etc. --- whether we speak of (our)
multigraphs, in which there exist exactly two types of relationships
between agents, or
ordinary graphs. In the case of multigraphs, 
%when using the mentioned
%terminology, 
we refer to their underlying ordinary graphs. %$G=(V,E)$. 
We also %identify 
use the same
terminology for 
$n\times n$ matrices $\mathbf{A}$ and their underlying graphs
$([n],\set{(i,j)\,|\, A_{ij}\neq 0})$.
%, i.e., we use the same
%terminology for both. 
%underlying the
%multigraphs $G=(V,\mathcal{E}=(E,m))$. 
%Moreover, since we treat
%operators 
%$\mathbf{W}\odot\mathbf{F}$ and the corresponding multigraphs as
%%`isomorphic', or, simply, 
%equivalent concepts, we also speak of
%$\mathbf{W}\odot\mathbf{F}$ as `strongly connected', etc. 
\end{remark}
%Although we do not need the concept of strongly connected components
%momentarily, it will be useful in follow-up sections, so that we
%define it here in the discussion of network properties.
%\begin{definition}
%  A \emph{strongly connected component} of a directed graph $G$ is a
%  maximal strongly connected subgraph, where a \emph{subgraph} of a
%  graph $G$ is a graph whose vertex set is a subset of that of $G$, and
%  whose edge relation is a subset of that of $G$ restricted to this
%  subset. 
%\end{definition}
\subsection{Matrix and Markov chain theory}\label{sec:matrix}
We %are ready to 
first state one of the main theorems for the DeGroot
updates \eqref{eq:matrix2} in the non-opposition case
(cf.\ \textcite{Golub2010}). We assume that 
$\mathbf{W}$ is row-stochastic. 
\begin{theorem}\label{th:trust1}
  Consider the opinion updating process \eqref{eq:matrix2} with
  $F_{ij}=\mathfrak{F}$ 
  for all $i,j\in[n]$, where $\mathfrak{F}$ is the identity function. Let the
  multigraph  
  corresponding to the operator
  $\mathbf{W}\odot\mathbf{F}=\mathbf{W}$ --- an ordinary graph --- be
  strongly connected and 
  aperiodic. Then $\mathbf{W}\odot\mathbf{F}$ is convergent and
  induces a consensus. 
\end{theorem}

In case $\mathbf{W}\odot\mathbf{F}$ is an affine-linear
map, whether or not $\mathbf{W}\odot\mathbf{F}$ %is a contraction
%mapping 
converges 
can be fully determined by reference the %well-known 
notion of
eigenvalues, which we introduce now.    
\begin{definition}
  Let $\mathbf{A}\in\real^{n\times n}$ be an $n\times n$ matrix. An
  eigenvalue of $\mathbf{A}$ is any value $\lambda\in \complex$ such
  that 
  %\begin{align*}
    $\mathbf{Ax} = \lambda\mathbf{x}$
  %\end{align*}
  for some non-zero vector $\mathbf{x}\in\real^n$. The set of distinct
  eigenvalues of matrix $\mathbf{A}$ is called its \emph{spectrum} 
  %of
  %matrix $\mathbf{A}$ 
  and denoted by $\sigma(\mathbf{A})$. By
  $\rho(\mathbf{A})$, we denote the \emph{spectral radius} of
  $\mathbf{A}$, the largest absolute value of all the eigenvalues of
  $\mathbf{A}$, that is, 
  %\begin{align*}
    $\rho(\mathbf{A}) = \max\set{\abs{\lambda}\sd
      \lambda\in\sigma(\mathbf{A})}$. 
  %\end{align*}
\end{definition}
\begin{theorem}[\cite{Meyer2000}, p.630]\label{theorem:spectral}
  For $\mathbf{A}\in\real^{n\times n}$, % was \complex before 
  $\lim_{t\goesto\infty}\mathbf{A}^t$ exists if and only if
  \begin{align*}
    &\rho(\mathbf{A})<1,\quad\text{or else},\\
    &\rho(\mathbf{A})=1 \text{ and $\lambda=1$ is the only eigenvalue
      on the unit circle, and $\lambda=1$ is semisimple},
  \end{align*}
  where an eigenvalue is called \emph{semisimple} if its algebraic
  multiplicity equals its geometric multiplicity. The \emph{algebraic
  multiplicity} of an eigenvalue $\lambda$ is the number of times it is
  repeated as a root of the characteristic polynomial
  $\chi(\lambda)=\det{(\mathbf{A}-\lambda\textbf{I}_n)}$, where
  $\mathbf{I}_n$ is the 
    $n\times n$ identity matrix. The \emph{geometric multiplicity} is
  the 
  number of linearly independent eigenvectors associated with
  $\lambda$. 
\end{theorem} 
\subsection{Signed networks}\label{sec:signed}
Here, we assume social networks $\mathbf{W}\odot\mathbf{F}$ such that
$F_{ij}\in\set{\mathfrak{F},\mathfrak{D}}$ where $\mathfrak{D}$ is
soft opposition on $S_0$, i.e., $\mathfrak{D}(x)=-x$. %As seen in
%Section \ref{sec:distrust}, 
Such operators admit a matrix
representation $\mathbf{A}$ in which each entry has a positive or
negative sign (or is zero), see
Proposition \ref{prop:consensusOpposition}. We assume that
$\mathbf{A}$ is connected and aperiodic. 
\begin{lemma}\label{lemma:dad}
  %Let $\mathfrak{D}$ be soft opposition on $S_0$. 
  %$S=[-\beta,\beta]$, for some
  %$\beta>0$. 
  Let $\mathbf{W}\odot\mathbf{F}$ be 
  %an arbitrary operator with
  %representation $\mathbf{A}$ 
  such that $A_{ii}=0$ and
  $A_{ij}=A_{ji}$. Then, $\mathbf{W}\odot\mathbf{F}$ is opposition
  bipartite if and only if there exists a diagonal matrix
  $\mathbf{\Delta}$ such that
  $\mathbf{\Delta}\mathbf{A}\mathbf{\Delta}=\length{\mathbf{A}}$,
  where $\length{\mathbf{A}}$ denotes the matrix with entries
  $\length{A_{ij}}$. 
\end{lemma}
%\begin{proof}
%  Let $\mathbf{W}\odot\mathbf{F}$ be opposition bipartite with
%  partition $(\struct{N}_1,\struct{N}_2)$. 
%  Choose $\Delta_{ii}=1$ if $i\in\struct{N}_1$ and $\Delta_{ii}=-1$ if
%  $i\in\struct{N}_2$. Then, as one can verify,
%  $\mathbf{\Delta}\mathbf{A}\mathbf{\Delta}=\length{\mathbf{A}}$. 
%
%  Conversely, let
%  $\mathbf{\Delta}\mathbf{A}\mathbf{\Delta}=\length{\mathbf{A}}$ so
%  that $\length{A_{ij}}=\Delta_{ii}\Delta_{jj}A_{ij}$. Hence, if
%  $A_{ij}\neq 0$, $\Delta_{ii},\Delta_{jj}\in\set{\pm 1}$. Choose
%  $i\in \struct{N}_1$ if $\Delta_{ii}=1$ and $i\in\struct{N}_2$
%  otherwise. %Thus, if $i,j\in\struct{N}_1$, then
%             %$\Delta_{ii}=\Delta_{jj}=1$
%             This is totally wrong: Choose $i,j$ in the same group if
%  and only if $\Delta_{ii}\Delta_{jj}\ge 0$.  
%\end{proof}
\begin{lemma}\label{lemma:similarityTransform}
  Let $\length{\mathbf{A}}=\mathbf{\Delta}\mathbf{A}\mathbf{\Delta}$
  as in Lemma \ref{lemma:dad}. Then $\mathbf{A}$ and
  $\length{\mathbf{A}}$ have the same eigenvalues with the same
  multiplicities. 
\end{lemma}
\begin{lemma}\label{lemma:1inA}
  %Let $\mathfrak{D}$ be soft opposition on $S_0$. 
  %$S=[-\beta,\beta]$, for some
  %$\beta>0$. 
  Let $\mathbf{W}\odot\mathbf{F}$ be %an arbitrary operator with
  %representation $\mathbf{A}$ 
  such that $A_{ii}=0$ and
  $A_{ij}=A_{ji}$. Then, $\mathbf{W}\odot\mathbf{F}$ is opposition
  bipartite if and only if $\lambda=1$ is an eigenvalue of
  $\mathbf{A}$. 
\end{lemma}
\begin{proof}
  \textcite{Altafini2013}, Lemma 1, shows that
  $0\in\sigma(\mathbf{L})$ if and only if $\mathbf{A}$ is opposition
  bipartite where $\mathbf{L}=\mathbf{I}-\mathbf{A}$. Clearly,
  $1\in\sigma(\mathbf{A})\iff 0\in\sigma(\mathbf{L})$. 
\end{proof}
\subsection{Proofs of main results}\label{sec:proofs}
      \begin{proposition}\label{prop:w-}
      Let
      $\mathbf{W}\odot \mathbf{F}$ be an arbitrary social
      network such that
      $F_{ij}\in \set{\mathfrak{F},\mathfrak{D}_i}$, for all 
      %where deviation function $\mathfrak{D}_i$ corresponds to agent
      $i,j\in[n]$. Then, for all $c\in S$, 
      \begin{align*}
      c\in\bigcap_{(i,j)\in [n]\times[n]}\text{Fix}(F_{ij}) \implies
      (c,\ldots,c)\in\text{Fix}(\mathbf{W}\odot\mathbf{F}). 
      \end{align*}
      Moreover, let $A=\set{i\in[n]\sd W_{i,\text{Out}(i)}>C_0}$
      denote the set of agents whose weight mass assigned to out-group
      members exceeds a particular threshold $C_0$; in the discrete
      case, $C_0=\frac{1}{2}$, and in the continuous case,
      $C_0=0$. Then, for any $i\in A$, it holds that
      \begin{align*}
        c\notin \text{Fix}(\mathfrak{D}_i)\implies 
        (c,\ldots,c)\notin \text{Fix}(\mathbf{W}\odot\mathbf{F}). 
      \end{align*}

      %In other words, 
      Combining both implications yields that 
      \begin{align*}
      P_1[\text{Fix}(\mathbf{W}\odot \mathbf{F})\cap \mathcal{C}]=\bigcap_{i\in
      A}\text{Fix}(\mathfrak{D}_i).
      \end{align*}
      %Let $c\in\text{Fix}(F_{ij})$ for all $(i,j)\in
      %[n]\times [n]$. Then, 
%
%      Conversely, 
      \end{proposition}
\begin{remark}
        %No analogous result when $\mathfrak{D}$ varies over $i,j$. 
        If $\mathfrak{D}$ is allowed to vary across both
        $i$ \emph{and} $j$, then
        $c\notin\text{Fix}(\mathfrak{D}_{ij})$ does not necessarily
        imply that
        $(c,\ldots,c)\notin\text{Fix}(\mathbf{W}\odot\mathbf{F})$. To
        see this, assume, for example, that in a three-player society
        $\set{1,2,3}$ agent $1$ has $\text{Out}(1)=\set{2,3}$ with
        $W_{12}=W_{13}=\frac{1}{4}$. %, and let $\text{Out}(2)=\text{Out}(3)=\emptyset$. 
        For a $c\in S=\real$, let
        $\mathfrak{D}_{12}(c)=c+\epsilon$ and let
        $\mathfrak{D}_{13}(c)=c-\epsilon$, for some $\epsilon>0$. Assuming that
        $\text{Out}(2)=\text{Out}(3)=\emptyset$, we have
        $(\mathbf{W}\odot\mathbf{F})\mathbf{c}=\mathbf{c}$, since, in
        particular, for agent $1$, 
        \begin{align*}
          W_{11}c+W_{12}\mathfrak{D}_{12}(c)+W_{13}\mathfrak{D}_{13}(c)=\frac{1}{2}c+\frac{1}{4}(c+\epsilon)+\frac{1}{4}(c-\epsilon)=c. 
        \end{align*}
\end{remark}      
\begin{proof}[Proof of Proposition \ref{prop:w-}]
      We only provide the proof for the continuous model. The discrete
      model proof is similar. 

        If $c=F_{ij}(c)$ for some $c\in S$ and all $(i,j)\in [n]\times
     [n]$, then clearly --- letting $\mathbf{c}=(c,\ldots,c)$ --- 
  $(\mathbf{W}\odot\mathbf{F})\mathbf{c}=\mathbf{c}$ by the definition
  of $\mathbf{W}\odot\mathbf{F}$ since for each agent $i\in[n]$, 
  \begin{align*}
    \bigl[(\mathbf{W}\odot\mathbf{F})\mathbf{c}\bigr]_{i}=
    %\sum_{j\in\text{In}_i}W_{ij}c+\sum_{j\in\struct{O}_i}W_{ij}D(c) =
    \sum_{j=1}^n W_{ij}F_{ij}(c) = 
    c\sum_{j\in[n]}W_{ij} = c = [\mathbf{c}]_i.
  \end{align*}

  Conversely, let $c\neq \mathfrak{D}_i(c)$ for some $c\in S$ and some
  $i\in A$. 
  %Let $i\in [n]$ be such
  %that $F_{ij}=D$ and $W_{ij}>0$ for some $j\in[n]$. 
  If
  $\mathbf{c}=(c,\ldots,c)$ were a fixed-point of
  $\mathbf{W} \odot\mathbf{F}$, then 
  \begin{align*}
    c = \sum_{j\in
      \text{Out}(i)}W_{ij}\mathfrak{D}_i(c)+\sum_{j\in\text{In}(i)}W_{ij}c =
    \mathfrak{D}_i(c)W_{i,\text{Out}(i)}+c(1-W_{i,\text{Out}(i)}),
  \end{align*}
  which implies that
  \begin{align*}
    %0 = W_{i,\struct{O}_i}(D(c)-c),
    W_{i,\text{Out}(i)}c = W_{i,\text{Out}(i)}\mathfrak{D}_i(c).
  \end{align*}
  This is a contradiction since $W_{i,\text{Out}(i)}>0$ by assumption.       
      \end{proof}
\begin{lemma}\label{lemma:w-}
        Let $\mathbf{W}\odot\mathbf{F}$ be an arbitrary social
        network. Assume that either $\mathbf{W}\odot\mathbf{F}$ refers
        to the discrete model or that each function $F_{ij}$ in
        $\mathbf{F}$ is continuous. 
        Then:
        \begin{align*}
                \text{Lim}(\mathbf{W}\odot\mathbf{F}) = 
                \text{Fix}(\mathbf{W}\odot\mathbf{F}).
        \end{align*}
\end{lemma}
\begin{proof}[Proof of Lemma \ref{lemma:w-}]
        The relation
        $\text{Fix}(\mathbf{W}\odot\mathbf{F})\subseteq \text{Lim}(\mathbf{W}\odot\mathbf{F})$
        is obvious. 
        Conversely, if each $F_{ij}$ is continuous, then
        $\mathbf{W}\odot\mathbf{F}$ is a continuous operator and thus,
        each limit vector
        $\mathbf{b}(\infty)\in\text{Lim}(\mathbf{W}\odot\mathbf{F})$ is a
        fixed-point of $\mathbf{W}\odot\mathbf{F}$:
        \begin{align*}
        (\mathbf{W}\odot\mathbf{F})\mathbf{b}(\infty)
        =
        (\mathbf{W}\odot\mathbf{F})\lim_{t\goesto\infty}(\mathbf{W}\odot\mathbf{F})^{t}\mathbf{b}(0) 
        = \lim_{t\goesto\infty}(\mathbf{W}\odot\mathbf{F})^{t+1}\mathbf{b}(0)
        = \mathbf{b}(\infty).
\end{align*}
         If $S$ is finite and $\mathbf{W}\odot\mathbf{F}$ is
convergent (for $\mathbf{b}(0)$), then $\mathbf{b}(\infty)$ is a
fixed-point of $\mathbf{W}\odot\mathbf{F}$ no matter the specification
of $\mathbf{F}$. 
        % by
        %Remark \ref{rem:fixed}. 
        %The same remark proves the assertion
        %in case $\mathbf{W}\odot\mathbf{F}$ refers to the discrete
        %model. 
\end{proof}
\begin{proof}[Proof of Theorem \ref{cor:w-}]
  This is an application of Proposition \ref{prop:w-} and
  Lemma \ref{lemma:w-}.  
\end{proof}
\begin{proposition}\label{prop:consensusOpposition}
  %Let $S=[\alpha,\beta]$ and let $D$ be soft opposition. 
  Let $\mathfrak{D}$ be of the form $ax+b$ for some constants $a$ and
  $b$, and let $F_{ij}\in\set{\mathfrak{F},\mathfrak{D}}$. 
  Then,
  $\mathbf{W}\odot\mathbf{F}$ is an affine-linear operator of the form
  $\mathbf{Ax}+\mathbf{d}$, that is,
  $(\mathbf{W}\odot\mathbf{F})(\mathbf{x})=\mathbf{A}\mathbf{x}+\mathbf{d}$ for
  all $\mathbf{x}\in S^n$. 
  %Moreover, 
  %if the spectral radius $\rho(\mathbf{A})$ of matrix $\mathbf{A}$
  %satisfies $\rho(\mathbf{A})<1$, then
  %$\mathbf{W}\odot\mathbf{F}$ induces the unique consensus
  %$(\frac{-b}{a-1},\ldots,\frac{-b}{a-1})^\intercal$, for all
  %initial opinion profiles 
  %$\mathbf{b}(0)\in S^n$.
\end{proposition}
\begin{proof}[Proof of Proposition \ref{prop:consensusOpposition}]
  %The proposition is clear, except maybe for the representation of
  %$\mathbf{W}\odot\mathbf{F}$ as an affine-linear operator. 
  For each agent
  $i\in[n]$, we have
  \begin{align*}
    \bigl[(\mathbf{W}\odot\mathbf{F})\mathbf{x}\bigr]_i =& \sum_{j\in
      \text{In}(i)} W_{ij}x_j +\sum_{j\in\text{Out}(i)} W_{ij}\mathfrak{D}(x_j) = \sum_{j\in
      \text{In}(i)} W_{ij}x_j +\sum_{j\in\text{Out}(i)}
    W_{ij}(ax_j+b)\\
    &= \sum_{j\in
      \text{In}(i)} W_{ij}x_j +\sum_{j\in\text{Out}(i)}
    aW_{ij}x_j + b\sum_{j\in\text{Out}(i)}W_{ij} \\&= \sum_{j\in
      \text{In}(i)} W_{ij}x_j +\sum_{j\in\text{Out}(i)}
    (aW_{ij})x_j + bW_{i,\text{Out}(i)}. 
  \end{align*}
  Thus, we can set $\mathbf{A}\in\real^{n\times n}$,
  $\mathbf{d}\in\real^n$ with
  \begin{align}\label{eq:A-opposition}
   A_{ij} = \begin{cases}
     W_{ij} & \text{if } F_{ij}=\mathfrak{F},\\
     aW_{ij} & \text{if } F_{ij}=\mathfrak{D},
     \end{cases}\quad
   \quad d_i = bW_{i,\text{Out}(i)}.
  \end{align}
\end{proof}
\begin{lemma}\label{lemma:inverse}
  Let $\mathbf{W}\odot\mathbf{F}$ be an arbitrary social network with
  $F_{ij}\in\set{\mathfrak{F},\mathfrak{D}}$ for an arbitrary
  deviation function $\mathfrak{D}$. Then,
  $\mathbf{W}\odot\mathbf{F}$ is opposition bipartite if and only if
  $\mathbf{W}\odot\bar{\mathbf{F}}$ is reverse opposition bipartite,
  where $\bar{\mathbf{F}}$ is the matrix with entries
  $\bar{F}_{ij}=\lnot F_{ij}$, whereby we define $\lnot
  \mathfrak{D}=\mathfrak{F}$ and $\lnot \mathfrak{F}=\mathfrak{D}$. 
 \end{lemma}
\begin{proof}
  %In a partition $(\struct{N}_1,\struct{N}_2)$
  See Figure \ref{fig:bipOpposition}, in Section \ref{sec:discrete}, for a
  graphical proof.  
\end{proof}
\begin{remark}
If $\mathfrak{D}$ is soft opposition on $S_0$, %=[-\beta,\beta]$ or $S=\real$, 
  let $(\mathbf{A},\mathbf{0})$ be the representation of
  $\mathbf{W}\odot\mathbf{F}$. Then, the lemma specializes to the
  statement that, in this situation, 
  $(\mathbf{A},\mathbf{0})$ is opposition bipartite if and only if
  $(-\mathbf{A},\mathbf{0})$ is reverse opposition bipartite. 
\end{remark}

\begin{proof}[Proof of Theorem \ref{theorem:main}]
  (i) If $\mathbf{W}\odot\mathbf{F}$ induces a (non-zero) polarization, then,
  necessarily, $1\in\sigma(\mathbf{A})$. But, $1\in
  \sigma(\mathbf{A})\iff \mathbf{W}\odot\mathbf{F}$ is opposition
  bipartite by Lemma \ref{lemma:1inA}. Conversely, let
  $\mathbf{W}\odot\mathbf{F}$ be opposition 
  bipartite. Then, $\length{\mathbf{A}}$ 
  and $\mathbf{A}$ are
  \emph{isospectral}, that is,
  they have the same eigenvalues and with the same associated %geometric
  multiplicities by Lemmas \ref{lemma:dad} and \ref{lemma:similarityTransform}. Now, ($\star$) a strongly connected and aperiodic
  row-stochastic matrix $\length{\mathbf{A}}$ has exactly one eigenvalue on the
  unit circle, $\lambda=1$, with algebraic and geometric multiplicity
  of $1$. Therefore, $\mathbf{A}$ has exactly one eigenvalue on the
  unit circle, $\lambda=1$, with algebraic and geometric multiplicity
  of $1$ and, consequently, converges by Theorem
  \ref{theorem:spectral}. Moreover, since each polarization
  vector $\mathbf{x}$ with $x_i=1$ if $i\in\struct{N}_1$ and $x_i=-1$
  if $i\in\struct{N}_2$ satisfies
  $\mathbf{Ax}=(\mathbf{W}\odot\mathbf{F})\mathbf{x}=\mathbf{x}$ when
  $\mathbf{W}\odot\mathbf{F}$ is opposition bipartite with partition
  $(\struct{N}_1,\struct{N}_2)$, $\mathbf{W}\odot\mathbf{F}$ induces a
  polarization that is not always zero (note that the geometric multiplicity 
  of $\lambda=1$ of $\mathbf{A}$ 
  is $1$). 

  Part
  (ii) ``$\Leftarrow$'' follows from the fact that $1\in\sigma(\mathbf{A})\iff $
  $\mathbf{W}\odot\mathbf{F}$ is opposition bipartite and the fact
  that $\mathbf{W}\odot\mathbf{F}$ with representation $\mathbf{A}$ is
  opposition bipartite if and only if $-\mathbf{A}$ is reverse opposition
  bipartite by Lemma \ref{lemma:inverse}. Thus,
  $-1\in\sigma(\mathbf{A})\iff $ 
  $\mathbf{W}\odot\mathbf{F}$ is reverse opposition bipartite,
  %and if
  %$\struct{G}(\length{\mathbf{A}})$ is strongly connected and
  %aperiodic, $\lambda=-1$ is thus the only eigenvalue of $\mathbf{A}$
  %on the unit circle, 
  whence $\mathbf{A}$ diverges by Theorem
  \ref{theorem:spectral}. Conversely, when $\mathbf{W}\odot\mathbf{F}$
  diverges, then $\rho(\mathbf{A})=1$ (we have $\rho(\mathbf{A})\le 1$ for all
  such matrices $\mathbf{A}$ as we consider since
  $|\mathbf{A}|=\mathbf{W}$ is row-stochastic and therefore,
  $\rho(\mathbf{A})\le \rho(|\mathbf{A}|)=1$). If $1$ were in
  $\sigma(\mathbf{A})$, then $\mathbf{W}\odot\mathbf{F}$ were
  opposition bipartite and $\mathbf{W}\odot\mathbf{F}$ would converge
  by (i). Hence, $1\notin\sigma(\mathbf{A})$ and consequently,
  $-1\in\sigma(\mathbf{A})$ --- since a symmetric matrix $\mathbf{A}$
  has no complex eigenvalues.  
  Consequently, $\mathbf{W}\odot\mathbf{F}$ is
  reverse opposition bipartite.  

  Finally, for part (iii), if $\mathbf{W}\odot\mathbf{F}$ is
  neither opposition bipartite nor reverse opposition bipartite, then, by
  our above reasonings, $\pm 1\notin\sigma(\mathbf{A})$, and since
  $\mathbf{A}$ is symmetric, $\mathbf{A}$ has no complex eigenvalues,
  whence $\rho(\mathbf{A})<1$. Thus, %by Theorem 
  %(3) $1\in\sigma(\mathbf{A})\iff $
  %$\mathbf{W}\odot\mathbf{F}$ is opposition bipartite. Thus, if
  %$\mathbf{W}\odot\mathbf{F}$ is not opposition bipartite, then, since
  %$\rho(\mathbf{A})\le 1$ by Proposition \ref{prop:radius}, either 
  %$\rho(\mathbf{A})<1$ or $\rho(\mathbf{A})=1$ but
  %$1\notin\sigma(\mathbf{A})$. If $\rho(\mathbf{A})<1$,
  $\mathbf{W}\odot\mathbf{F}$ 
  %is a contraction mapping and %whence 
  %consequently 
  %$\mathbf{W}\odot\mathbf{F}$ 
  induces 
  the unique neutral consensus
  $(0,\ldots,0)$. % by Banach's fixed-point theorem,
  %, Theorem
  %\ref{theorem:banach}, 
  %where $c=0$ due to the choice of
  %$\mathfrak{D}$. 
  Conversely, if $\mathbf{W}\odot\mathbf{F}$ induces the
  neutral consensus $(0,\ldots,0)$ for each initial belief vector
  $\mathbf{b}(0)$, then necessarily $\rho(\mathbf{A})<1$. 
  %(if
  %$\rho(\mathbf{A})$ were equal to $1$, some non-zero vector
  %$\mathbf{x}$ would satisfy $\mathbf{A}\mathbf{x}=(-)\mathbf{x}$), 
  Hence, 
  $\mathbf{W}\odot\mathbf{F}$ is neither opposition bipartite nor
  reverse opposition  
  bipartite. 
  %If
  %$\rho(\mathbf{A})=1$ but $1\notin\sigma(\mathbf{A})$,
  %$\mathbf{W}\odot\mathbf{F}$ does not converge by Theorem
  %\ref{theorem:spectral}. 

  Now, fact ($\star$) is a classical theorem for row-stochastic matrices,
  which is, e.g., based on the famous Perron-Frobenius theorem; in our
  context, it is given by combining Theorems \ref{th:trust1} and
  \ref{theorem:spectral}, for example. %We prove 
  %Facts (1) and (2) are based on 
  %%the appendix, 
  %Lemmas \ref{lemma:dad},
  %\ref{lemma:similarityTransform}, and \ref{lemma:1inA}, respectively.  
  %$\mathbf{\Delta}\mathbf{A}\mathbf{\Delta}$, which is nonnegative
  %row-stochastic, has the same spectrum as $\mathbf{A}$ if
  %$\mathbf{W}\odot\mathbf{F}$ is opposition bipartite. Hence, since
  %the graph underlying $\mathbf{W}\odot\mathbf{F}$ is strongly
  %connected and aperiodic, $\mathbf{\Delta}\mathbf{A}\mathbf{\Delta}$
  %has an underlying graph structure that is strongly connected and
  %aperiodic as well. Thus, $\mathbf{\Delta}\mathbf{A}\mathbf{\Delta}$
  %has exactly one simple eigenvalue on the unit circle,
  %$\lambda=1$. Since $\mathbf{\Delta}\mathbf{A}\mathbf{\Delta}$ and
  %$\mathbf{A}$ are isospectric, $\mathbf{W}\odot\mathbf{F}$ converges
  %for all initial opinions $\mathbf{b}(0)$; and it converges to a
  %polarization. 

  %Conversely, if $\mathbf{W}\odot\mathbf{F}$ is not opposition
  %bipartite, $\mathbf{W}\odot\mathbf{F}$ has no eigenvalue $\lambda=1$
  %and does, thus, either not converge or, if $\rho(\mathbf{A})<1$,
  %converges to $\mathbf{0}$. 
\end{proof}
\begin{proof}[Proof of Theorem \ref{theorem:social_influence}]
  (i) $\implies$ (ii): 
  For an appropriate initial opinion vector, 
  let each agent hold limiting opinions $a\neq 0$ or $-a$ as
  indicated. Place agents in a group $\struct{N}_1$
  resp. $\struct{N}_2$ depending on whether they hold limit opinions
  $a$ or $-a$, respectively. We show that
  $(\struct{N}_1,\struct{N}_2)$ forms an opposition bipartite
  partition of $[n]$. Take $i,i'\in\struct{N}_1$ and assume that
  $A_{ii'}<0$ (i.e., $i$ and $i'$ are enemies). Then, since the limit
  opinion vector is a fixed-point of $\mathbf{A}$, we have: 
  \begin{align*}
    a(\pm A_{i1}+\cdots+A_{ii'}+\cdots+\pm A_{in}) = a.
  \end{align*}
  But this cannot be, since 
  %$\sum_{j}\length{A_{ij}}=1$, it follows that ${\pm
  %  A_{i1}+\cdots+A_{ii'}+\cdots+\pm A_{in}}<1$, since 
  $A_{ii'}+C<\length{A_{ii'}}+C'=1$ where
  $C=\sum_{j\neq i'}\pm A_{ij}$ and $C'=\sum_{j\neq
    i'}\length{A_{ij}}$. %, a contradiction.  
  Similarly, we can show that no two agents 
  $i\in\struct{N}_1$, $j\in\struct{N}_2$ are friends of each
  other. Hence, $\mathbf{W}\odot\mathbf{F}$ is opposition
  bipartite. 
  %If $\mathbf{W}\odot\mathbf{F}$ were in addition periodic,
  %then $\mathbf{W}\odot\mathbf{F}$ would diverge for some initial
  %opinion vector by Theorem \ref{theorem:mainmain}, a contradiction. 

  (ii) $\implies$ (i): 
  Conversely, let $\mathbf{W}\odot\mathbf{F}$ be opposition bipartite
  and aperiodic. 
  By Lemma \ref{lemma:dad}, $\mathbf{W}\odot\mathbf{F}$ is opposition
  bipartite if and only if there exists a diagonal matrix
  $\mathbf{\Delta}$ (with entries $\pm 1$) such that
  $\length{\mathbf{A}}=\mathbf{\Delta}\mathbf{A}\mathbf{\Delta}$. We
  know that
  \begin{align*}
    \lim_{t\goesto\infty}\length{\mathbf{A}}^t\mathbf{p} = \mathbf{s}^\intercal\mathbf{p}\mathbbm{1}
  \end{align*}
  for all $\mathbf{p}\in S^n$ 
  (see, e.g., \cite[Proposition 1]{Golub2010}) 
  and that
  $\inv{\mathbf{\Delta}}=\mathbf{\Delta}$. 
  Therefore,
  \begin{align*}
    \lim_{t\goesto\infty} \mathbf{A}^t\mathbf{p} =
    \mathbf{\Delta}\lim_{t\goesto\infty}
    \length{\mathbf{A}}^t(\mathbf{\Delta}\mathbf{p}) = \mathbf{\Delta}\mathbf{s}^\intercal(\mathbf{\Delta}\mathbf{p})\mathbbm{1}. 
  \end{align*}
  This proves the theorem. 
\end{proof}

%%%%%%%%%%%%%%%%%%%%%%%%%%%%%%%%%%%
%%%%%%%%%%%%%%%%%%%%%%%%%%%%%%%%%%%
%%%%%%%%%%%%%%%%%%%%%%%%%%%%%%%%%%%
%\begin{theorem}[Gershgorin circle theorem]
%  Let $\mathbf{A}\in\real^{n\times n}$ be an $n\times n$ matrix. Then,
%  \begin{align*}
%    \sigma(\mathbf{A})\subseteq \bigcup_{i\in [n]}B(A_{ii};\sum_{j\neq i}\abs{A_{ij}}),
%  \end{align*}
%  where $B(x;r)$ denotes the closed ball centered around $x$ with
%  radius $r$.
%\end{theorem}

\end{appendices}

\printbibliography

\end{document}